\documentclass[11pt]{article}

\usepackage[utf8]{inputenc}
\usepackage[T1]{fontenc}
\usepackage[sc]{mathpazo}
\usepackage{microtype}
\usepackage{amsmath}
\usepackage{amssymb}
\usepackage{amsthm}
\usepackage{bm}
\usepackage{dsfont}
\usepackage{authblk}
\usepackage{fullpage}
\usepackage{comment}
\usepackage{tikz}
\usepackage{mathtools}
\usepackage[shortlabels]{enumitem}
\usepackage{complexity}
\usepackage[backend=biber, style=alphabetic, maxbibnames=99, url=false]{biblatex}
\usepackage{hyperref}
\usepackage[capitalize, nameinlink]{cleveref}
\usepackage{nag}

\newcommand{\declarecolor}[2]{\definecolor{#1}{RGB}{#2}\expandafter\newcommand\csname #1\endcsname[1]{\textcolor{#1}{##1}}}
\declarecolor{White}{255, 255, 255}
\declarecolor{Black}{0, 0, 0}
\declarecolor{LightGray}{216, 216, 216}
\declarecolor{Gray}{127, 127, 127}
\declarecolor{Orange}{237, 125, 49}
\declarecolor{LightOrange}{251,229, 214}
\declarecolor{Yellow}{255, 192, 0}
\declarecolor{LightYellow}{255, 242, 200}
\declarecolor{Blue}{91, 155, 213}
\declarecolor{LightBlue}{222, 235, 247}
\declarecolor{Green}{112, 173, 71}
\declarecolor{LightGreen}{226, 240, 217}
\declarecolor{Navy}{68, 114, 196}
\declarecolor{LightNavy}{218, 227, 243}

\hypersetup{
	colorlinks=true,
	pdfpagemode=UseNone,
	citecolor=Green,
	linkcolor=Navy,
	urlcolor=Black,
	pdfstartview=FitW
}

\newcommand{\declareperson}[1]{\expandafter\newcommand\csname#1\endcsname[1]{\textcolor{Orange}{#1: ##1}}}

\declareperson{Nima}
\declareperson{Kuikui}
\declareperson{Shayan}
\declareperson{Cynthia}

\theoremstyle{plain}
\newtheorem{theorem}{Theorem}[section]
\newtheorem{lemma}[theorem]{Lemma}
\newtheorem{corollary}[theorem]{Corollary}
\newtheorem{proposition}[theorem]{Proposition}
\newtheorem{fact}[theorem]{Fact}

\theoremstyle{definition}
\newtheorem{definition}[theorem]{Definition}

\theoremstyle{remark}
\newtheorem{remark}[theorem]{Remark}

\newlist{parts}{enumerate}{10}
\setlist[parts]{label=\arabic*.,ref=\arabic*}
\makeatletter
\if@cref@capitalise
	\crefname{partsi}{Part}{Parts}
\else
	\crefname{partsi}{part}{parts}
\fi
\makeatother
\Crefname{partsi}{Part}{Parts}

\usetikzlibrary{shapes, patterns, decorations, fit, intersections}

\newcommand*{\Z}{{\mathbb{Z}}}
\let\R\relax
\newcommand*{\R}{{\mathbb{R}}}

\newcommand*{\bone}{{\mathds{1}}}
\newcommand*{\cM}{{\mathcal{M}}}
\newcommand*{\eps}{{\epsilon}}
\newcommand*{\cI}{{\mathcal{I}}}
\newcommand*{\cB}{{\mathcal{B}}}
\newcommand*{\tP}{{\tilde{P}}}
\newcommand*{\Hess}[1]{\nabla^2 #1}
\newcommand*{\SharpP}{{\ComplexityFont{\#P}}}

\DeclareMathOperator{\supp}{supp}

\DeclareMathOperator{\rank}{rank}
\DeclareMathOperator{\diag}{diag}

\DeclareMathOperator{\expansion}{h}
\DeclareMathOperator{\cond}{cond}

\DeclareMathOperator{\vol}{vol}

\providecommand{\given}{}
\DeclarePairedDelimiterX{\card}[1]{\lvert}{\rvert}{\renewcommand\given{\nonscript\:\delimsize\vert\nonscript\:\mathopen{}}#1}
\DeclarePairedDelimiterX{\abs}[1]{\lvert}{\rvert}{\renewcommand\given{\nonscript\:\delimsize\vert\nonscript\:\mathopen{}}#1}
\DeclarePairedDelimiterX{\norm}[1]{\lVert}{\rVert}{\renewcommand\given{\nonscript\:\delimsize\vert\nonscript\:\mathopen{}}#1}
\DeclarePairedDelimiterX{\tuple}[1]{\lparen}{\rparen}{\renewcommand\given{\nonscript\:\delimsize\vert\nonscript\:\mathopen{}}#1}
\DeclarePairedDelimiterX{\parens}[1]{\lparen}{\rparen}{\renewcommand\given{\nonscript\:\delimsize\vert\nonscript\:\mathopen{}}#1}
\DeclarePairedDelimiterX{\brackets}[1]{\lbrack}{\rbrack}{\renewcommand\given{\nonscript\:\delimsize\vert\nonscript\:\mathopen{}}#1}
\DeclarePairedDelimiterX{\set}[1]\{\}{\renewcommand\given{\nonscript\:\delimsize\vert\nonscript\:\mathopen{}}#1}
\let\Pr\relax
\DeclarePairedDelimiterXPP{\Pr}[1]{\mathbb{P}}[]{}{\renewcommand\given{\nonscript\:\delimsize\vert\nonscript\:\mathopen{}}#1}
\DeclarePairedDelimiterXPP{\PrX}[2]{\mathbb{P}_{#1}}[]{}{\renewcommand\given{\nonscript\:\delimsize\vert\nonscript\:\mathopen{}}#2}
\DeclarePairedDelimiterXPP{\Ex}[1]{\mathbb{E}}[]{}{\renewcommand\given{\nonscript\:\delimsize\vert\nonscript\:\mathopen{}}#1}
\DeclarePairedDelimiterXPP{\ExX}[2]{\mathbb{E}_{#1}}[]{}{\renewcommand\given{\nonscript\:\delimsize\vert\nonscript\:\mathopen{}}#2}
\DeclarePairedDelimiter{\dotprod}{\langle}{\rangle}

\title{Log-Concave Polynomials II: High-Dimensional  Walks and an FPRAS for Counting Bases of a Matroid}

\author{Nima Anari}
\affil{\small Stanford University, \textsf{anari@cs.stanford.edu}}

\author{Kuikui Liu}
\author{Shayan Oveis Gharan}
\affil{\small University of Washington, \textsf{liukui17@cs.washington.edu}, \textsf{shayan@cs.washington.edu}}

\author{Cynthia Vinzant}
\affil{\small North Carolina State University, \textsf{clvinzan@ncsu.edu}}

\bibliography{refs}

\begin{document}
\maketitle
	
\begin{abstract}
    We design an FPRAS to count the number of bases of any matroid given by an independent set oracle, and to estimate the partition function of the random cluster model of any matroid in the regime where $0<q<1$. Consequently, we can sample random spanning forests in a graph and (approximately) compute the reliability polynomial of any matroid. We also prove the thirty year old conjecture of Mihail and Vazirani that the bases exchange graph of any matroid has expansion at least 1.
    
 Our algorithm and the proof build on the recent results of Dinur, Kaufman, Mass and Oppenheim \cite{KM17,DK17, KO18} who show that high dimensional walks on simplicial  complexes  mix rapidly if the corresponding localized random walks on 1-skeleton of links of all complexes
  are strong spectral expanders. One of our key observations is a close connection between pure simplicial complexes and multiaffine homogeneous polynomials. Specifically, if $X$ is a pure simplicial complex with positive weights on its maximal faces, we can associate with $X$ a multiaffine homogeneous polynomial $p_{X}$ such that the eigenvalues of the localized random walks on $X$ correspond to the eigenvalues of the Hessian of derivatives of $p_{X}$.
\end{abstract}
\section{Introduction}
	Let $\mu:2^{[n]}\to \R_{\geqslant 0}$ be a probability distribution on the subsets of the set $[n]=\set{1,2,\dots,n}$. We assign a multiaffine polynomial with variables $x_1,\dots,x_n$ to $\mu$,
	\[ g_\mu(x) = \sum_{S\subseteq [n]} \mu(S)\cdot \prod_{i\in S} x_i.\]
	The polynomial $g_\mu$ is also  known as the \emph{generating polynomial} of $\mu$. A polynomial $p\in\R[x_1,\dots,x_n]$ is $d$-homogeneous if every monomial of $p$ has degree $d$. We say $\mu$ is  \emph{$d$-homogeneous} if the polynomial $g_\mu$ is $d$-homogeneous, meaning that $\card{S}=d$ for any $S$ with $\mu(S)>0$.

	A polynomial $p \in \R[x_1,\dots,x_n]$ with nonnegative coefficients is \emph{log-concave} on a subset $K\subseteq \R_{\geqslant 0}^n$ if $\log p$ is a concave function at any point in $K$, or equivalently, its Hessian $\nabla^{2}\log p$ is negative semidefinite on $K$. We say a polynomial $p$ is \emph{strongly log-concave} on $K$ if for any $k\geqslant 0$, and any sequence of integers $1\leqslant i_1,\dots,i_k\leqslant n$, 
	\[ (\partial_{i_1} \cdots \partial_{i_k} p)(x_1,\dots,x_n)\]
	is log-concave on $K$. In this paper, for convenience and clarity, we only work with (strong) log-concavity with respect to the all-ones vector, $\bone$. So, unless otherwise specified, $K=\set{\bone}$ in the above definition. We say the distribution $\mu$ is \emph{strongly log-concave} at $\bone$ if $g_\mu$ is strongly log-concave at $\bone$. The notion of strong log-concavity was first introduced by \textcite{Gur09,Gur10} to study approximation algorithms for mixed volume and multivariate generalizations of Newton's inequalities.

	In this paper,  we show that the ``natural'' Monte Carlo Markov Chain (MCMC) method on the support of a $d$-homogeneous strongly log-concave  distribution $\mu:2^{[n]}\to \R_{\geqslant 0}$ mixes rapidly. This chain can be used to generate random samples from a distribution arbitrarily close to $\mu$.

	The chain $\cM_{\mu}$ is defined as follows. We take the state space of $\cM_\mu$ to be the support of $\mu$, namely $\supp(\mu) = \set{S\subseteq [n]\given \mu(S)\neq 0}$. For  $\tau\in \supp(\mu)$, first we drop an element $i \in \tau$, chosen uniformly at random from $\tau$.  Then, among all sets $\sigma \supset \tau \setminus \set{i}$ in the support of $\mu$ we choose one with probability proportional to $\mu(\sigma)$.

	It is easy to see that $\cM_{\mu}$ is reversible with stationary distribution $\mu$. Furthermore, assuming $g_\mu$ is strongly log-concave, we will see that $\cM_\mu$ is irreducible. We prove that this chain mixes rapidly. More formally, for a state $\tau$ of the Markov chain $\cM$, and $\epsilon > 0$, the \emph{total variation mixing time} of $\cM$ started at $\tau$ with transition probability matrix $P$ and stationary distribution $\pi$ is defined as follows:
	\[ t_{\tau}(\epsilon)=\min\set{t\in \Z_{\geqslant 0}\given \norm{P^t(\tau,\cdot)-\pi}_1 \leqslant \epsilon},\]
 	where $P^t(\tau,\cdot)$ is the distribution of the chain started at $\tau$ at time $t$.

	The following theorem is the main result of this paper.
	\begin{theorem}\label{thm:SLCmixing}
		Let $\mu:2^{[n]} \rightarrow \R_{\geqslant0}$ be a $d$-homogeneous strongly log-concave probability distribution. If $P_{\mu}$ denotes the transition probability matrix of $\cM_{\mu}$ and $X(k)$ denotes the collection of size-$k$ subsets of $[n]$ which are contained in some element of $\supp(\mu)$, then for every $0 \leqslant k \leqslant d-1$, $P_{\mu}$ has at most $\card{X(k)} \leqslant \binom{n}{k}$ eigenvalues of value $>1 - \frac{k+1}{d}$. In particular, $\cM_{\mu}$ has spectral gap at least $1/d$, and if $\tau$ is in the support of $\mu$ and $0<\epsilon<1$, the total variation mixing time of the Markov chain $\cM_{\mu}$ started at $\tau$ is at most
		\[ t_\tau(\eps) \  \leqslant \ d\log\parens*{\frac1{\epsilon\mu(\tau)}}. \]
	\end{theorem}
	To state the key corollaries of this theorem, we will need the following definition.
	\begin{definition}
		Given a domain set $\Omega$, which is compactly represented by, say, a membership oracle, and a nonnegative weight function $w:\Omega \rightarrow \R_{\geqslant0}$, a \emph{fully polynomial-time randomized approximation scheme (FPRAS)} for computing the partition function $Z = \sum_{x \in \Omega} w(x)$ is a randomized algorithm that, given an error parameter $0 < \epsilon < 1$ and error probability $0 < \delta < 1$, returns a number $\tilde{Z}$ such that ${\rm Prob}[(1 - \epsilon)Z \leqslant \tilde{Z} \leqslant (1 + \epsilon)Z] \geqslant 1 - \delta$. The algorithm is required to run in time polynomial in the problem input size, $1/\epsilon$, and $\log(1/\delta)$.
	\end{definition}
	Equipped with this definition, we can now concisely state the main applications of \cref{thm:SLCmixing}. \cref{thm:SLCmixing} gives us an algorithm to efficiently sample from a distribution which approximates $\mu$ closely in total variation distance. By the equivalence between approximate counting and approximate sampling for self-reducible problems \cite{JVV86}, this gives an FPRAS for each of the following:
	\begin{enumerate}
	    \item counting the bases of a matroid, and
	    \item estimating the partition function of the random cluster model for a new range of parameter values
	\end{enumerate}
	For real linear matroids, we also give an algorithm for estimating the partition function of a generalized version of a $k$-determinantal point process.
	Note that these problems are all instantiations of the following: estimate the partition function of some efficiently computable nonnegative weights on bases of a matroid. Furthermore, as the restriction and contraction of a matroid by a subset of the ground set are both (smaller) matroids, problems of this form are indeed self-reducible. In the following sections we discuss these applications in greater depth.


\subsection{Counting Problems on Matroids}
	Let $M=([n],\cI)$ be an arbitrary matroid on $n$ elements (see \cref{subsec:matroids}) of rank $r$. Let $\mu$ be the uniform distribution on the bases of the matroid $M$. It follows that $\mu$ is $r$-homogeneous. Using the Hodge-Riemann relation proved by \textcite{AHK18}, a subset of the authors proved \cite{AOV18} that for any matroid $M$, $\mu$ is strongly log-concave.\footnote{Indeed, in \cite{AOV18}, it is shown that $g_{\mu}$ satisfies a seemingly stronger property known as ``complete log-concavity'', namely that $\partial_{v_{1}} \dotsb \partial_{v_{k}} g_M$ is log-concave (at $\bone$) for any sequence of \emph{directional} derivatives $\partial_{v_{1}}\dotsb \partial_{v_{k}}$ with nonnegative directions $v_{1},\dots,v_{k} \in \R_{\geqslant0}^{n}$. We will prove in a future companion paper that complete log-concavity is equivalent to strong log-concavity.} This implies that the chain $\cM_{\mu}$ converges rapidly to stationary distribution. This gives the first polynomial time algorithm to generate a uniformly random base of a matroid. Note that to run $\cM_{\mu}$ we only need an oracle to test whether a given set $S\subseteq [n]$ is an independent set of $M$. Therefore, with only polynomially many queries (in $n,r,\log(1/\eps)$) we can generate a random base of $M$.
	\begin{corollary}
		For any matroid $M=([n],\cI)$ of rank $r$, any basis $B$ of $M$ and $0<\eps<1$, the mixing time of the Markov chain $\cM_{\mu}$ 
		starting at $B$ is at most
		\[ t_B(\eps) \ \leqslant \ r\log(n^{r}/\epsilon) \leqslant  \  r^2 \log(n/\eps).\]
	\end{corollary}
	To prove this we simply used the fact that a matroid of rank $r$ on $n$ elements has at most $\binom{n}{r}\leqslant n^r$ bases. There are several immediate consequences of the above corollary. Firstly, by equivalence of approximate counting and approximate sampling for self-reducible problems \cite{JVV86} we can count the number of bases of any matroid given by an independent set oracle up to a $1+\eps$ multiplicative error in polynomial time.
	\begin{corollary}
		There is a randomized algorithm that for any matroid $M$ on $n$ elements with rank $r$ given by an independent set oracle, and any $0<\eps<1$, counts the number of bases of $M$ up to a multiplicative factor of $1\pm\eps$ with probability at least $1-\delta$ in time polynomial in $n,r,1/\eps,\log(1/\delta)$.
	\end{corollary}
	As an immediate corollary for any $1\leqslant k\leqslant r$ we can count the number of independent sets of $M$ of size $k$. This is because if we truncate $M$ to independent sets of size at most $k$ it remains a matroid. As a consequence we can generate uniformly random forests in a given graph, and compute the reliability polynomial
	\begin{align*}
	    C_{M}(p) = \sum_{S \subseteq [n] : \rank(S) = r} (1 - p)^{\card{S}} p^{n - \card{S}}
	\end{align*}
	for any matroid and $0 \leqslant p \leqslant 1$, all in polynomial time. Note this latter fact follows from the ability to count the number of independent sets of a fixed size, as the complements of rank-$r$ subsets $S \subseteq [n]$ are precisely the independent sets of the dual of $M$. Prior to our work, we could only compute the reliability polynomial for graphic matroids due to a recent work of Guo and Jerrum \cite{GJ18a}. 
	
	One can associate a graph $G_M$ to any matroid $M$, called \emph{the bases exchange} graph. This graph has a vertex for every basis of $M$ and two bases $B,B'$ are connected by an edge if $\card{B\Delta B'}=2$. It follows by the bases exchange property of matroids that this graph is connected. For an unweighted graph $G=(V,E)$, the expansion of a set $S\subset V$ and the graph $G$ are defined as 
	\[\expansion(S)=\frac{\card{E(S,\overline{S})}}{\card{S}} 
	\ \ \ \ \text{ and } \ \ \ \
	\expansion(G)=\min_{S:\card{S}\leqslant \card{\overline{S}}} \expansion(S).
	\]
	\Textcite{MV89} conjectured that the bases exchange graph has expansion at least one, i.e., that $\expansion(G_M)\geqslant 1$, for any matroid $M$. It turns out that the bases exchange graph is closely  related to the Markov chain $\cM_{\mu}$. The following theorem is an immediate consequence of the above corollary.
	\begin{theorem}\label{thm:basesexchange}
	For any matroid $M$, the expansion of the bases exchange graph is at least $1$, $\expansion(G_M)\geqslant 1$.	
	\end{theorem}

	\subsection{The Random Cluster Model}\label{subsec:RandomCluster}
	Another application of this theory is estimating the partition function of the random cluster model. For a matroid $M = ([n],\cI)$ of rank $r$ and parameters $p, q$, the partition function of the random cluster model from statistical mechanics due to Fortuin and Kasteleyn \cite{For71, FK72I, For72II, For72III} is the following polynomial function associated to $M$,
	\begin{align*}
	Z_M(p,q)=\sum_{S\subseteq [n]} q^{r+1-\rank(S)} p^{\card{S}}
	\end{align*}
	where $\rank(S)$ is the size of the largest independent set contained in $S$. We note that typically one scales each term by $(1-p)^{n-\card{S}}$ but up to a normalization factor (and change of variables) the two polynomials are equivalent. We refer interested readers to a recent book of  \textcite{Grim09} for further information. Typically, one considers the special case where $M$ is a graphic matroid, in which case the exponent of $q$ is simply the number of connected components of $S$. To the best of our knowledge, prior to this work, one could only compute $Z_M$ when $q=2$ because of the close connection to the Ising model \cite{JS93, GJ17}. Our next result is a polynomial time algorithm that estimates $Z_{M}(p,q)$ for any $0< q\leqslant 1$ and $p\geqslant 0$.

	\begin{theorem}\label{thm:randomclusterSLC}
		For a matroid $M$ with rank function $\rank:2^{[n]}\rightarrow \Z_{\geqslant 0}$, parameter $0 < q \leqslant 1$ and choice of ``external field'' $\bm{\lambda} = (\lambda_{1},\dots,\lambda_{n}) \in \R_{>0}^{n}$, the polynomial
		\begin{align*}
		    f_{M,k,q}(x_{1},\dots,x_{n}) = \sum_{S \in \binom{[n]}{k}} q^{-\rank(S)} \prod_{i\in S}\lambda_{i}x_i
		\end{align*}
		is strongly log-concave.
	\end{theorem}
	Together with \cref{thm:SLCmixing}, this gives an FPRAS for estimating $f_{M,k,q}(\bone)$ given an independence oracle for the matroid $M$. Estimating $Z_{M}(p, q)$ then follows as
	\begin{align*}
	    Z_{M}(p,q) = q^{r+1}\sum_{k=0}^{n} p^{k}f_{M,k,q}(\bone)
	\end{align*}
	and each term is nonnegative. In fact, the polynomial $Z_{M}$ is closely related to the Tutte polynomial
	\begin{align*}
	    T_{M}(x,y) = \sum_{S \subseteq [n]} (x - 1)^{r - \rank(S)}(y - 1)^{\card{S} - \rank(S)}.
	\end{align*}
	Indeed, we can write
	\begin{align*}
	    T_{M}(x,y) = \frac{1}{(x-1)(y-1)^{r+1}}Z_{M}\parens*{y-1,(x-1)(y-1)}
	\end{align*}
	Hence, an FPRAS for estimating $Z_{M}(p,q)$ for $p \geqslant 0$ and $0 \leqslant q \leqslant 1$ gives an FPRAS for estimating $T_{M}(x, y)$ in the region described by the inequalities $y \geqslant 1$ and $0 \leqslant (x - 1)(y - 1) \leqslant 1$.



	\subsection{Determinantal Distributions on Real Linear Matroids}
	Finally, we show that the class of homogeneous multiaffine strongly log-concave polynomials is closed under raising all coefficients to a fixed exponent less than $1$.
	\begin{theorem}\label{thm:cpow}
	Let $f= \sum_{S \subseteq [n]} c_{S}\prod_{i\in S}x_i$ be a homogeneous degree-$k$ multiaffine strongly log-concave polynomial. Then $f_{\alpha} = \sum_{S \subseteq [n]} c_{S}^{\alpha} \prod_{i\in S}x_i$ is strongly log-concave for every $0 \leqslant \alpha \leqslant 1$.
	\end{theorem}

	We use the above theorem to design a sampling algorithm for determinantal point processes. A \emph{determinantal point process (DPP)} on a set of elements $[n]$ is a probability distribution $\mu:2^{[n]}\rightarrow \R_{\geqslant 0}$ identified by a positive semidefinite matrix $L \in  \R^{n\times n}$ where for any $S\subseteq [n]$ we have
	\[\mu(S)  \ \propto \  \det(L_S),\]
	where $L_S$ is the principal sub-matrix of $L$ indexed by the elements of $S$. Determinantal point processes are fundamental to the study of a variety of tasks in machine learning, including text summarization, image search, news threading, and diverse feature selection \cite[see, e.g.,][]{KT12}. A \emph{$k$-determinantal point process ($k$-DPP)}  is a determinantal point process conditioned on the sets $S$ having size $k$.

	Given a positive semidefinite matrix $L$, let $\mu$ be the corresponding $k$-DPP. We have
	\[ g_\mu(x)\propto \sum_{S\in \binom{[n]}{k}} \det(L_S)\cdot \prod_{i\in S}x_i.\]
	It turns out that the above polynomial is real stable and so it is strongly log-concave over $\R^n_{\geqslant 0}$ \cite[see, e.g.,][]{AOR16}. \Textcite{AOR16} show that a natural Markov chain with the Metropolis rule mixes rapidly and generates a random sample of $\mu$. The above theorem immediately implies the following log-concavity result.
	\begin{corollary}\label{cor:dpp}
		For every positive semidefinite matrix $L \succcurlyeq 0$ and exponent $0 \leqslant \alpha \leqslant 1$, the polynomial
		\begin{align*}
		    \sum_{S\in \binom{[n]}{k}} \det(L_S)^{\alpha} \prod_{i\in S}x_i
		\end{align*}
		is strongly log-concave.
	\end{corollary}
	It follows from \cref{thm:SLCmixing} that for any $0\leqslant\alpha\leqslant 1$ we can generate samples from a ``smoothed'' $k$-DPP distribution, where for any set $S$, $\Pr{S}\propto \det(L_S)^\alpha$, in polynomial time. The weights $\det(L_{S})^{\alpha}$ may be thought of as a way to interpolate between two extremes for selecting diverse data points.

	We also note that for $\alpha = 1/2$, it is known that \cref{cor:dpp} follows from the Brunn-Minkowski theorem applied to appropriately defined zonotopes. For $\alpha = 0$ when the $k$-DPP has full support, and for $\alpha = 1$ as mentioned earlier, the above polynomial is actually real stable, and hence strongly log-concave. \Cref{thm:cpow} gives a unified proof that all of these polynomials are strongly log-concave.

\subsection{Related Works}
	There is a long line of work on designing approximation algorithms to count the  bases of a matroid.  Most of these works focus on expansion properties of bases exchange graph. \Textcite{FM92} showed that for a special class of matroids known as \emph{balanced matroids} \cite{MS91,FM92}, the bases exchange graph has expansion at least 1. A matroid $M$ is balanced if for any minor of $M$ (including $M$ itself), the uniform distribution over its bases satisfies the pairwise negative correlation property. Many of the extensive results in this area \cite{Gam99,JS02,JSTV04,Jur06,Clo10,CTY15,AOR16} only study approximation algorithms for this limited class of matroids, and not much is known beyond the class of balanced matroids. Unfortunately, many interesting matroids are not balanced. An important example is the  matroid of all acyclic subsets of edges of a graph $G=(V,E)$ of size at most $k$ (for some $k<\card{V}-1$) \cite{FM92}.

	There has been other approaches for counting bases. \Textcite{GJ18b} used the popping method to count bases of bicircular matroids. \Textcite{BS07} designed a \emph{randomized} algorithm that gives, roughly, a $\log(n)^r$ approximation factor to the number of bases of a  given matroid with $n$ elements and  rank $r$. In \cite{AOV18}, a subset of the authors gave a deterministic $e^r$ approximation to the number of bases using the fact that $g_\mu(M)$ is  log-concave over $\R^n_{\geqslant 0}$. 

	There is an extensive literature on hardness of exact computation and inapproximability of the Tutte polynomial and the partition function of the random cluster model. It is known that exact computation of the Tutte polynomial for a graph is \SharpP-hard at all points $(x,y)$ except at $(1,1), (-1,-1), (0,-1), (-1,0)$, along the hyperbola $(x-1)(y-1) = 1$, and for planar graphs, along the hyperbola $(x-1)(y-1) = 2$ \cite{JVW90}, \cite{Ver91}, \cite{Wel94}. In the realm of inapproximability, it is known that even for planar graphs, there is no FPRAS to approximate the Tutte polynomial for $x > 1, y < -1$ or $y > 1, x < -1$ assuming \NP{} $\neq$ \RP{} \cite{GJ08, GJ12II}. Furthermore, there is no FPRAS for estimating the partition function $Z_{M}$ of the random cluster model on general graphic matroids when $q > 2$, nor is there an FPRAS for $Z_{M}$ at $q = 2$ for general binary matroids, unless there is an FPRAS for counting independent sets in a bipartite graph \cite{GJ12I,GJ13,GJ14}.

\subsection{Independent Work} 
	In a closely related upcoming work, Br\"and\'en and Huh, in a slightly different language,  independently prove the strong log-concavity of several of the polynomials that appear in this paper. 
	In upcoming papers, both groups of authors use these techniques to prove the strongest form of Mason's conjecture and further study closure properties of (strongly) log-concave polynomials.

\subsection{Techniques}
    One of our key observations is a close connection between pure simplicial complexes and multiaffine homogeneous polynomials. Specifically, if $X$ is a pure simplicial complex with positive weights on its maximal faces, we can associate with $X$ a multiaffine homogeneous polynomial $p_{X}$ such that the eigenvalues of the localized random walks on $X$ correspond to the eigenvalues of the Hessian of derivatives of $p_{X}$.
    \begin{table}[htb]
        \centering
        \begin{tabular}{|c|c|}\hline
             Weighted Simplicial Complex $X$ & Multiaffine Polynomial $p_{X}$ \\ \hline
             Dimension-$d$ & Degree-$d$ \\
             Weight of $\emptyset$ & Evaluation at $\bone$\\
             Connectivity of Links & Indecomposability \\
             Link & Differentiation\\
             Local Random Walk & (Normalized) Hessian\\ \hline
        \end{tabular}
        
        \label{tab:my_label}
    \end{table}
    Using this correspondence, one can study multiaffine homogeneous polynomials using techniques from simplicial complexes, and vice versa. To study the walk $\mathcal{M}_{\mu}$ corresponding to a polynomial $g_{\mu}$, we analyze the simplicial complex corresponding to $g_{\mu}$. To do this, we leverage recent developments in the area of high-dimensional expanders, which we discuss below.

	Given a simplicial complex $X$ (see \cref{sec:SimplicialComplex}) and an ordering of its vertices, one can associate a high dimensional Laplacian matrix to the $k$-dimensional faces of $X$. These matrices generalize the classical graph Laplacian and there has been extensive research to study their eigenvalues \cite[see][and the references therein]{Lub17}. A method known as Garland's method \cite{Gar73} relates the eigenvalues of graph Laplacians of 1-skeletons of links of $X$ to eigenvalues of high dimensional Laplacians of $X$  \cite[see][]{BS97, Opp18}.

	Recently, \textcite{KM17} studied a high dimensional walk on a simplicial complex, which is closely related to the walk $\cM_\mu$ that we defined above (see \cref{sec:highdimwalk}). Their goal is to argue that, similar to classical expander graphs, high dimensional walks mix rapidly on a high dimensional expander. Their bounds were improved in a work of \textcite{DK17}, who showed that if all nontrivial eigenvalues of the simple random walk matrix on all 1-skeletons of links of $X$ have absolute value at most $\lambda$, then the high dimensional walk on $k$-faces of $X$ has spectral gap at least $\frac1{k+2}-O((k+1)\lambda)$. This was further improved in a recent work of \textcite{KO18}: They showed that if all non-trivial eigenvalues  of the simple random walk matrix on all 1-skeleton of links of $X$ are at most $\lambda$, then the spectral gap of the high dimensional walk is at least $\frac1{k+2}-(k+1)\lambda$. In other words, negative eigenvalues of the random walk matrix do not matter. One  only needs positive eigenvalues to be small. 

	Note that in order to make the spectral gap bounds meaningful one needs $\lambda \ll \frac1{k^2}$. In other words, one needs that, except the trivial eigenvalue of 1, all other eigenvalues are either negative or very close to $0$. Here is the place where the connection to (strong) log-concavity comes into the picture. A polynomial $p$ is log-concave at $\bone$ if $\nabla^2 p(\bone)$ has at most one positive eigenvalue. A polynomial $p$ is strongly log-concave if the same holds for all partial derivatives of $p$. Our main observation is that this property is equivalent to taking $\lambda=0$ in the corresponding simplicial complex. Namely, we obtain the best possible spectral gap of $\frac{1}{k+2}$ when the simplicial complex comes from a strongly log-concave polynomial.

	Our approach has a close connection to the original plan of \textcite{FM92} who used the negative correlation property of balanced matroids to show that the bases exchange walk mixes rapidly. Unfortunately, most interesting matroids do not satisfy negative correlation. But it was observed \cite{AHK18,HW17,AOV18} that all matroids satisfy a \emph{spectral} negative dependence property. Namely, consider the uniform distribution $\mu$ over the bases of a matroid $M$, and consider the Hessian $\nabla^2 \log(g_{\mu})$ of the $\log$ of the generating polynomial $g_{\mu}$ at the point $x =\bone$. Then $M$ is negatively correlated if and only if all off-diagonal entries of this matrix are non-positive, whereas $M$ being spectrally negatively correlated means that this matrix is negative semidefinite. Spectral negative correlation is precisely what one needs to bound the mixing time of the high dimensional walk on the corresponding simplicial complex.

\paragraph{Structure of the paper.} 
	In \cref{sec:prelim} we discuss necessary background on linear algebra, matroids, simplicial complexes and strongly log-concave polynomials. We also provide a useful characterization of strong log-concavity. In \cref{sec:highdimwalk} we discuss and reprove a version of the main theorem of \textcite{KO18} on mixing time of high dimensional walks,  \cref{thm:localexpander}. In \cref{sec:SLCtoLSE} we use this to prove \cref{thm:SLCmixing} and the Mihail-Vazirani conjecture, \cref{thm:basesexchange}.  Finally, in \cref{sec:applications} we first prove our new characterization of strong log-concavity and discuss its applications. Specifically, we give a self-contained proof that the uniform distribution over the bases of a matroid is strongly log-concave and we prove \cref{thm:randomclusterSLC,thm:cpow}.

\paragraph{Acknowledgements.} 
Part of this work was started while the first and last authors were visiting the Simons Institute for the Theory of Computing. It was partially supported by the DIMACS/Simons Collaboration on Bridging Continuous and Discrete Optimization through NSF grant CCF-1740425. Shayan Oveis Gharan and Kuikui Liu are supported by the NSF grant CCF-1552097 and ONR-YIP grant N00014-17-1-2429. Cynthia Vinzant was partially supported by the National Science Foundation grant DMS-1620014. 

We thank Lap Chi Lau, Mark Jerrum and Alan Frieze for helpful comments on an earlier version of this manuscript.

\section{Preliminaries}\label{sec:prelim}
	First, let us establish some notational conventions. Unless otherwise specified, all logarithms are in base $e$. All vectors are assumed to be column vectors. For two vectors $\phi, \psi\in \R^n$, we use $\dotprod{\phi, \psi}$ to denote the standard Euclidean inner product between $\phi$ and $\psi$. We use $\R_{>0}$ and $\R_{\geqslant 0}$ to denote the set of positive and nonnegative real numbers, respectively, and $[n]$ to denote $\set{1,\dots,n}$. For a vector $x\in \R^n$ and a set $S\subseteq [n]$, we let $x^S$ denote $\prod_{i\in S} x_i$.

	We use $\partial_{x_i}$ or $\partial_i$ to denote the partial differential operator $\partial/\partial x_i$. We denote the gradient of a function or polynomial $p$ by $\nabla p$ and the Hessian of $p$ by $\Hess{p}$. 

\subsection{Linear Algebra}
	We say a matrix $A\in\R^{n\times n}$ is \emph{stochastic} if all entries of $A$ are nonnegative and every row adds up to exactly 1. It is well-known that the largest eigenvalue in magnitude of any stochastic matrix is $1$ and its corresponding eigenvector is the all-ones vector, $\bone$. If the eigenvalues of a matrix $A\in\R^{n\times n}$ are all real, then we order them as
	\[ \lambda_n(A) \leqslant \hdots \leqslant \lambda_1(A). \]
	A symmetric matrix $A\in\R^{n\times n}$ is positive semidefinite (PSD), denoted $A\succcurlyeq 0$, if all its eigenvalues $\lambda_k(A)$ are nonnegative, or equivalently if for all $v\in\R^n$,
	\[ v^\intercal Av \geqslant 0.\]
	Similarly, $A$ is negative semidefinite (NSD), denoted $A\preccurlyeq 0$, if $v^\intercal Av \leqslant 0$ for all $v\in\R^n$.
	Equivalently, a real symmetric matrix is PSD (NSD) if its eigenvalues are nonnegative (nonpositive), respectively. 
	\begin{theorem}[{Schur Product Theorem \cite[Thm~7.5.3]{HJ13}}]\label{thm:schur}
	If $A,B \in \R^{n \times n}$ are positive semidefinite, then their Hadamard product $A \circ B$, whose entries are $(A \circ B)_{i,j} = A_{i,j}B_{i,j}$, is positive semidefinite.
	\end{theorem}

	\begin{theorem}[{Perron-Frobenius Theorem \cite[Ch. 8]{HJ13}}]\label{thm:perron}
		Let $A \in \R^{n \times n}$ be symmetric and have strictly positive entries. Then $A$ has an eigenvalue $\lambda$ which is strictly positive. Furthermore, it has multiplicity one and its corresponding eigenvector $v$ has strictly positive entries.
	\end{theorem}

	\begin{theorem}[{Cauchy's Interlacing Theorem \cite[Corollary~4.3.9]{HJ13}}]\label{thm:CauchyInterlacing}
		For a symmetric matrix $A\in\R^{n\times n}$ and vector $v\in \R^n$, the eigenvalues of $A$ interlace the eigenvalues of $A+vv^\intercal$.
		That is, for  $B = A+vv^\intercal,$
		\[\lambda_n(A)\leqslant \lambda_n(B)\leqslant\lambda_{n-1}(A) \leqslant \dots \leqslant \lambda_{2}(B) \leqslant \lambda_{1}(A)  \leqslant \lambda_{1}(B). \]
	\end{theorem}	
	The following is an immediate consequence:
	\begin{lemma}\label{lem:Cauchy1eigenvalue}
		Let $A\in\R^{n\times n}$ be a symmetric matrix and let $P\in\R^{m\times n}$. If $A$ has at most one positive eigenvalue, then $PAP^\intercal$ has at most one positive eigenvalue.
	\end{lemma}
	\begin{proof}
		Since $A$ has at most one positive eigenvalue, we can write $A=B+vv^\intercal$ for some vector $v\in \R^n$ and some negative semidefinite matrix $B$. Then $PAP^\intercal = PBP^\intercal + Pvv^\intercal P^\intercal$. First, observe that $PBP^\intercal\preccurlyeq 0$, since for $x\in\R^m$, $x^\intercal PB P^\intercal x = (P^\intercal x)^\intercal B (P^\intercal x)\leqslant 0$. Second, let $w = Pv\in \R^m$. Then $Pvv^\intercal P^\intercal=ww^\intercal$ and by \cref{thm:CauchyInterlacing}, the eigenvalues of $PBP^\intercal$ interlace the eigenvalues of $PBP^\intercal+(Pv)(Pv)^\intercal$. Since all eigenvalues of $PBP^\intercal$ are nonpositive, $PAP^\intercal=PBP^\intercal+ww^\intercal$ has at most one positive eigenvalue. 
	\end{proof}

	The following fact is well-known.
	\begin{fact}\label{fact:eigenvaluesinvert}
		Let $A\in \R^{n\times k}$ and $B\in \R^{k\times n}$ be arbitrary matrices. Then, non-zero eigenvalues of $AB$ are equal to non-zero eigenvalues of $BA$ with the same multiplicity. 
	\end{fact}

	\begin{lemma}\label{lem:stochastic}
		Let $A\in \R^{n\times n}$ be a symmetric matrix with at most one positive eigenvalue. Then, for any PSD matrix $B\in\R^{n\times n}$, $BA$ has at most one positive eigenvalue.
	\end{lemma}
	\begin{proof}
	Since $B\succcurlyeq 0$, we can write $B=C^\intercal C$ for some $C\in \R^{n\times n}$. By \cref{fact:eigenvaluesinvert}, $BA=C^\intercal C A$ has the same nonzero eigenvalues as the matrix $C A C^\intercal$. Since $A$ has at most one positive eigenvalue, by \cref{lem:Cauchy1eigenvalue}, $C A C^\intercal$ has at most one positive eigenvalue and so does $BA$.
	\end{proof}

	\begin{lemma}\label{lem:adjacency1poseig}
		Let $A\in\R^{n\times n}$ be a symmetric matrix with nonnegative entries and at most one positive eigenvalue, and let $w(i)=\sum_{j=1}^n A_{i,j}$. Then,
		\[ A\preccurlyeq \frac{ww^\intercal}{\sum_i w(i)}.\]
	\end{lemma}
	\begin{proof}
		Let $W=\diag(w)$. Then, by \cref{lem:Cauchy1eigenvalue}, $\mathcal{A}=W^{-1/2} A W^{-1/2}$ has at most one positive eigenvalue. Observe that the top eigenvector of $\mathcal{A}$ is the $\sqrt{w}$ vector, where $\sqrt{w}(i) = \sqrt{w(i)}$, for all $i$. In particular, $\mathcal{A}\sqrt{w}=\sqrt{w}$. So, $\sqrt{w}$ is the only eigenvector  of $\mathcal{A}$ with positive eigenvalue and we have 
		\[ {\cal A} \preccurlyeq \frac{\sqrt{w}\sqrt{w}^\intercal}{\norm{\sqrt{w}}^2} = \frac{\sqrt{w}\sqrt{w}^\intercal}{\sum_i w(i)}. \]
		Multiplying both sides of the inequality on the left and right by $W^{1/2}$ proves the lemma.
	\end{proof}


	In this paper, we will often switch between different inner products. As such, we highlight the following variational characterization of eigenvalues of a linear operator that is self-adjoint with respect to an arbitrary inner product. In particular, the matrix of the linear operator need not be symmetric.
	\begin{theorem}[Courant-Fischer Theorem]\label{thm:Courant-Fischer}
		Let $T:\R^{n} \rightarrow \R^{n}$ be a linear operator that is self-adjoint with respect to some inner product $\dotprod{\cdot,\cdot}$ (not necessarily Euclidean). If $\lambda_{n} \leqslant \dots \leqslant \lambda_{1}$ are the eigenvalues of $T$, then
		\begin{align*}
		    \lambda_{k} &= \min_{U } \  \max_{v } \ \dotprod{v, Tv},
		\end{align*}
		where the minimum is taken over all $(n-k)$-dimensional subspaces $U\subseteq \R^n$ and the maximum is 
		taken over all vectors $v\in U$ with $\dotprod{v,v}= 1$.
	\end{theorem}
	When the inner product $\dotprod{\cdot,\cdot}$ is clear, we call a matrix $A$ self-adjoint when $\dotprod{u, Av}=\dotprod{Au, v}$,
	for all $u, v$. Similarly we call a self-adjoint $A$ positive semidefinite when for all $v$
	\[ \dotprod{v, Av}\geqslant 0. \]
	By \cref{thm:Courant-Fischer}, this is equivalent to $A$ having nonnegative eigenvalues.

	\subsection{Markov Chains and Random Walks}\label{subsec:randomwalks}
	For this paper, we consider a Markov chain as a triple $(\Omega,P,\pi)$ where  $\Omega$ denotes the (finite) state space, $P\in \R_{\geqslant 0}^{\Omega\times \Omega}$ denotes the transition probability matrix and $\pi \in \R_{\geqslant 0}^{\Omega}$ denotes a stationary distribution of the chain (which will be unique for all chains we consider). For  $\tau,\sigma \in \Omega$, we use $P(\tau,\sigma)$ to denote the corresponding entry of $P$, which is the probability of moving from $\tau$ to $\sigma$. 
	We say a Markov chain is $\epsilon$-lazy if for any state $\tau \in \Omega$, $P(\tau,\tau) \geqslant \epsilon$. A chain $(\Omega, P, \pi)$ is \emph{reversible} if there is a nonzero nonnegative function $f:\Omega\to\R_{\geqslant0}$ such that for any pair of states $\tau, \sigma \in \Omega$, 
	\[f(\tau) P(\tau,\sigma) = f(\sigma)P(\sigma,\tau).\]
	If this condition is satisfied, then $f$ is proportional to a stationary distribution of the chain. In this paper we only work with reversible Markov chains. Note that being reversible means that the transition matrix $P$ is self-adjoint w.r.t.\ the following $\dotprod{\cdot, \cdot}$ defined for $\phi,\psi\in \R^\Omega$:
	\[ \dotprod{\phi,\psi}=\sum_{x\in\Omega} f(x)\phi(x)\psi(x).  \]

	Reversible Markov chains can be realized as random walks on weighted graphs. Given a weighted graph $G=(V,E,w)$ where every edge $e\in E$ has weight $w(e)$, the non-lazy \emph{simple} random walk on $G$ is the Markov chain that from any vertex $u\in V$ chooses an edge $e=\set{u,v}$ with probability proportional to $w(e)$ and jumps to $v$. We can make this walk $\epsilon$-\emph{lazy} by staying at every vertex with probability $\epsilon$. It turns out that if $G$ is connected, then the walk has a unique stationary distribution where $\pi(u)\propto w(u)$, where $w(u)=\sum_{v\sim u} w(\set{u,v})$ is the weighted degree of $u$.

	For any reversible Markov chain $(\Omega, P, \pi)$, the largest eigenvalue of $P$ is $1$. We let $\lambda^*(P)$ denote the second largest eigenvalue of $P$ in absolute value. That is, if $-1\leqslant \lambda_n\leqslant \dots\leqslant \lambda_1=1$ are the eigenvalues of $P$, then $\lambda^*(P)$ equals $\max\set{\abs{\lambda_2},\abs{\lambda_n}}$.
	\begin{theorem}[{\cite[Prop 3]{DS91}}]\label{thm:mixingtime}
		For any reversible irreducible  Markov chain $(\Omega, P, \pi)$,  $\eps>0$, and any starting state $\tau\in \Omega$,
	\[ t_\tau(\eps)  \ \leqslant  \ \frac1{1-\lambda^*(P)}\cdot \log\parens*{\frac{1}{\eps\cdot \pi(\tau)}}.\]
	\end{theorem}
	For our results, it will be enough to look at the second largest eigenvalue $\lambda_2(P)$, which we can bound using the conductance of a weighted graph. Consider a weighted graph $G=(V,E,w)$ and a subset $S\subseteq V$ of vertices. We let $\overline{S}$ denote the complement $V\setminus S$. Then the \emph{conductance} of $S$, denoted by $\cond(S)$, is defined as 
	\[ \cond(S) \ = \ \frac{w(E(S,\overline{S}))}{\vol(S)} \ = \ \frac{\sum_{e\in E(S,\overline{S})} w(e)}{\sum_{v\in S} w(v)},\]
	where $E(S,\overline{S})=\set{ \set{u,v}\in E\given u\in S,v\notin S}$ is the set of edges between $S$ and $\overline{S}$, $w(E(S,\overline{S}))$ is the sum of weights of these edges, and the volume $\vol(S)$ is the sum of the weighted degrees of the vertices in $S$. The conductance of $G$ is then
	\[ \cond(G)\ = \ \min_{S} \  \cond(S), \]
	where the minimum is taken over subsets $\emptyset \subsetneq S\subsetneq V$ for which $w(S)\leqslant w(\overline{S})$. 

	We say $G$ is $d$-regular if $w(v)=d$ for all $v\in V$.
	\begin{theorem}[Cheeger's Inequalities \cite{AM85,Alon86}]\label{thm:Cheeger}
		For any $d$-regular weighted graph $G=(V,E,w)$, 
		\[ \frac{d-\lambda_2(A_G)}{2} \ \leqslant  \ \cond(G) \  \leqslant \  \sqrt{2(d-\lambda_2(A_{G}))}, \]
		where $A_G$ is the weighted adjacency matrix of $G$ given by $(A_G)_{ij} = w(\set{i,j})$. 
	\end{theorem}

	A direct consequence of the above theorem is that if the (weighted) graph $G$ is connected, i.e., for all proper nonempty subsets $S$ of vertices, $w(E(S,\overline{S}))>0$, then $\lambda_2(A_G)<d$. If the matrix $A_G$ is stochastic, then the graph is $1$-regular, which gives the following.

	\begin{corollary}\label{cor:cheegerconnected}
		If $A\in \R^{n\times n}$ is a stochastic matrix corresponding to a reversible Markov chain with the property that $\sum_{\card{\set{i,j}\cap S}=1} A_{i,j}>0$ for all subsets $\emptyset \subsetneq S\subsetneq [n]$, then $\lambda_2(A)<1$.
	\end{corollary}

\subsection{Matroids}\label{subsec:matroids}
	A \emph{matroid} $M=([n],\cI)$ is a combinatorial structure consisting of a ground set $[n]$ of elements and a nonempty collection $\cI$ of \emph{independent} subsets of $[n]$ satisfying: 
	\begin{enumerate}[i)]
		\item If $S\subseteq T$ and $T\in \cI$, then $S\in \cI$ (hereditary property).
	    \item If $S,T\in \cI$ and $\card{T}>\card{S}$, then there exists an element $i\in T \setminus S$ such that $S\cup \set{i}\in \cI$ (exchange axiom).
	\end{enumerate}
	The \emph{rank}, denoted by $\rank(S)$, of a subset $S \subset [n]$ is the size of any maximal independent set of $M$ contained in $S$. Thus, the independent sets of $M$ are precisely those subsets $S \subset [n]$ for which $\rank(S) = \card{S}$. We call $\rank([n])$ the rank of $M$, and if $M$ has rank $r$, any set $S\in \cI$ of size $r$ is called a \emph{basis} of $M$.

	An element $i \in [n]$ is a \emph{loop} if $\set{i} \notin \cI$, that is, $\set{i}$ is dependent. Two non-loops $i,j \in [n]$ are \emph{parallel} if $\set{i,j} \notin \cI$, that is, $\set{i,j}$ is dependent.

	\begin{definition}[Contraction]
	Let $M = ([n],\cI$) be a matroid and $S \in \cI$. Then the \emph{contraction $M/S$} is the matroid with ground set $[n] \setminus S$ and independent sets $\set{T \subseteq [n] \setminus S \given T \cup S \in \cI}$.
	\end{definition}

	We will use a key property of matroids called the \emph{matroid partition property}. For any matroid $M=([n],\cI)$,  the non-loops of $M$ can be partitioned into sets $S_1,S_2,\dots,S_k$ for some $1 \leqslant k\leqslant n$ with the property that non-loops $j,k \in [n]$ are parallel if and only if they belong to the same set $S_i$. Indeed, one can check from the axioms for a matroid that being parallel defines an equivalence relation on the non-loop elements of $[n]$ and $S_{1},\dots,S_{k}$ are then the corresponding equivalence classes.

\subsection{Simplicial Complexes}\label{sec:SimplicialComplex}

	A \emph{simplicial complex} $X$ on the ground set $[n]$ is a nonempty collection of subsets of $[n]$ that is downward closed, namely  if $\tau \subset \sigma$ and $\sigma \in X$, then $\tau \in X$. The elements of $X$ are called \emph{faces/simplices}, and the \emph{dimension} of a face $\tau \in X$ is defined as $\dim(\tau) = \card{\tau}$. Note that for convenience and clarity of notation, our definition deviates from the standard definition of $\dim(\tau) = \card{\tau}-1$ used by topologists.
	
	The empty set $\emptyset$ has dimension-$0$. A face of dimension 1 is a \emph{vertex} of $X$ and a face of dimension 2 is called an \emph{edge}. More generally, we write 
	\[X(k) \ = \ \set{\tau\in X\given \dim(\tau) =k}\]
	for the collection of dimension-$k$ faces, or \emph{$k$-faces/$k$-simplices}, of $X$. The dimension of $X$ is the largest $k$ for which $X(k)$ is nonempty, and we say that $X$ is \emph{pure} of dimension $d$ if all maximal faces of $X$ have the dimension $d$. In this paper we will only consider pure simplicial complexes.

	The \emph{link} of a face $\tau\in X$ denoted by $X_\tau$ is the simplicial complex on $[n]\setminus \tau$ obtained by taking all faces in $X$ that contain $\tau$ and removing $\tau$ from them,
	\[ X_\tau \ = \ \set{\sigma\setminus\tau\given  \sigma\in X, \sigma\supset\tau}.\]
	Note that if $X$ is pure of dimension $d$ and $\tau\in X(k)$, then $X_\tau$ is pure and has dimension $(d-k)$.

	For any matroid $M = ([n],\cI)$ of rank $r$, the independent sets $\cI$ form a pure $r$-dimensional simplicial complex on $[n]$ called its \emph{independence (or matroid) complex}. Furthermore, for any $S \in \cI$, the link $\cI_{S}$ of the independence complex consists precisely of the independent sets of the contraction $M/S$. There are many other beautiful simplicial complexes associated to matroids, but here we will be mainly interested in the independence complex. 

	We can equip a simplicial complex with a weight function: $w:X\to\R_{>0}$ which assigns a positive weight to each face of $X$. We say that $w$ is \emph{balanced} if for every non-maximal face $\tau\in X$ of dimension $k$, 
	\begin{align}
	    w(\tau) \ = \  \sum_{\sigma \in X(k+1) : \sigma \supset \tau} w(\sigma)\label{eq:balancedweight}
	\end{align}
	For a pure simplicial complex $X$ we can define a balanced weight function by assigning arbitrary positive weights to maximal faces and defining the weight of each lower dimensional face recursively. Indeed, if $X$ is a pure simplicial complex of dimension $d$ and $w$ is a balanced weight function, then, for any $\tau\in X(k)$,
	\[ w(\tau)=(d-k)! \sum_{\sigma\in X(d): \sigma\supset \tau} w(\sigma).\]
	One natural choice is the function which assigns a weight of one to each maximal face, but there are many other interesting choices. 
 
	Any balanced weight function on $X$ induces a weighted graph on the vertices of $X$ as follows. The \emph{1-skeleton} of $X$ is the graph on vertices $X(1)$ with edges $X(2)$.  Then, restricting $w$ to $X(1)$ and $X(2)$ determines a weighted graph, where $w(v)$ gives the weighted degree of each $v \in X(1)$. The weighted graphs coming from both $X$ and its links $X_\tau$ will be useful in later sections.

\subsection{Log-Concave Polynomials}
	We say a polynomial $p\in\R[x_1,\dots,x_n]$ is $d$-homogeneous if every monomial of $p$ has degree $d$; equivalently, $p$ is $d$-homogeneous if $p(\lambda x_1,\dots,\lambda x_n)=\lambda^d p(x_1,\dots,x_n)$ for every $\lambda \in \R$. For a $d$-homogeneous polynomial $p$, the following identity, known as Euler's identity, holds:
	\begin{align}
    	d\cdot p(x) &= \sum_{k=1}^{n} x_{k}\partial_k p(x). \label{eq:eulerppartial}
    \end{align}
	Note that if $p$ is homogeneous then all directional derivatives of $p$ are also homogeneous, so one can apply this to  $\partial_i p(x) $ and $\partial_i\partial_j p(x)$ to find that 
	\begin{align}
	\label{eq:Hesssum}
	    (d-1) \cdot (\nabla p) = \sum_{k=1}^{n} x_{k} \cdot (\nabla (\partial_k p)) = (\nabla^{2}p) \cdot x \quad \text{ and } \quad (d-2) \cdot (\nabla^{2}p) = \sum_{k=1}^{n} x_{k} \cdot (\nabla^{2}(\partial_k p)).
	\end{align}

	A polynomial $p \in \R[x_1,\dots,x_n]$ with nonnegative coefficients is log-concave if $\log p$ is a concave function over $\R_{>0}^{n}$. For simplicity we also consider the zero polynomial to be log-concave. Equivalently, $p$ is log-concave if the Hessian of $\log p$
	\[\nabla^2\log p=\frac{p\cdot (\nabla^2 p) - (\nabla p) (\nabla p)^\intercal}{p^{2}}\]
	is negative semi-definite at any point $x \in \R_{>0}^{n}$, where $\nabla p$ is the gradient of $p$. Since $(\nabla p)(\nabla p)^\intercal$ is a rank-1 matrix, by Cauchy's interlacing theorem, $p\cdot (\nabla^2 p)$ has at most one positive eigenvalue at any $x\in\R_{>0}^{n}$. Since $p$ has nonnegative coefficients and $x$ has strictly positive entries, $p(x)>0$ so $\nabla^{2} \log p$ being negative semidefinite is equivalent to $\nabla^{2}p \preccurlyeq \frac{(\nabla p)(\nabla p)^{\intercal}}{p}$, where the right-hand side is a rank-1 positive semidefinite matrix. In particular, $\nabla^{2}p$ has at most one positive eigenvalue at $x$. In \cite{AOV18} it is shown that for homogeneous polynomials $p$, the converse of this is also true, i.e., if $\nabla^2 p$ has at most one positive eigenvalue at all $x>0$ then $p$ is log-concave.
	\begin{proposition}[\cite{AOV18}]\label{prop:oneposeigenvalue}
		A degree-$d$ homogeneous polynomial $p\in\R[x_1,\dots,x_n]$ with nonnegative coefficients is log-concave over $\R_{>0}^{n}$ iff $(\nabla^2 p)(x)$ has at most one positive eigenvalue at all $x\in \R_{>0}^{n}$.
	\end{proposition}
	\begin{remark}
	Whenever $p$ has degree at least 2 and nonnegative coefficients, $(\nabla^{2}p)(x)$ has at least one positive entry for any $x \in \R_{>0}^{n}$. Hence, one can see (via, for example, the variational characterization of eigenvalues \cref{thm:Courant-Fischer} and using a test vector with positive entries) that $(\nabla^{2}p)(x)$ must have at least one strictly positive eigenvalue. Thus, whenever we write ``at most one positive eigenvalue'', we also mean it has ``exactly one positive eigenvalue''.
	\end{remark}


	We say a polynomial $p\in \R[x_1, \dots, x_n]$ is \emph{decomposable} if it can be written as a sum of polynomials in disjoint subset of the variables, that is, if there exists a nonempty subset $I \subsetneq [n]$ and nonzero polynomials $g\in \R[x_i : i\in I]$, $h \in \R[x_i : i\not\in I]$  for which $f = g+h$. We call $f$ \emph{indecomposable} otherwise.

	\begin{lemma}\label{prop:LCtoDecomp}
		If $p\in\R[x_1,\dots,x_n]$ has nonnegative coefficients, is homogeneous of degree at least 2, and log-concave at $\bone$, then $p$ is indecomposable. 
	\end{lemma}
	\begin{proof}
		Suppose that $p$ has nonnegative coefficients, is homogeneous of degree $\geqslant 2$, and is decomposable, with decomposition  $p = g+h$ where $g\in \R[x_i : i\in I]$ and $h \in \R[x_i : i\not\in I]$. Both $g$ and $h$ are restrictions of $p$ obtained by setting some variables equal to zero, therefore both $g$ and $h$ are log-concave. Then, at $\bone$, the Hessians of $g$ and $h$ each have precisely one positive eigenvalue.  However, the Hessian of $p$ at this point is a block diagonal matrix with these two blocks, $\nabla^2g,\nabla^2h$, 
			\[ \nabla^2 p= \begin{bmatrix} \nabla^2 g & 0 \\ 0& \nabla^2h \end{bmatrix}.\]
		 So, $p$ has exactly two positive eigenvalues, meaning that $p$ is not log-concave, a contradiction.
	\end{proof}

	In order to prove several distributions of interest are strongly log-concave, we will prove an equivalent characterization of strongly log-concave polynomials.
	\begin{theorem}\label{thm:SLCdeg2}
		Let $p\in \R[x_1,\dots,x_n]$ be a $d$-homogeneous polynomial such that:
		\begin{enumerate}
		    \item for any $0\leqslant k\leqslant d-2$ and any $(i_{1},\dots,i_{k}) \in [n]^k$, $\partial_{i_{1}}\dotsb \partial_{i_{k}} p$ is indecomposable, and 
		    \item for any  $(i_1,\dots,i_{d-2})\in [n]^{d-2}$, the quadratic $\partial_{i_1}\dots \partial_{i_{d-2}} p$ is either identically zero, or log-concave at $\bone$.
		\end{enumerate}
		Then $p$ is strongly log-concave at $\bone$. 
	\end{theorem}
	In \cref{prop:LCtoDecomp} we show that the condition that all partial derivatives are indecomposable is necessary for a polynomial to be (strongly) log-concave.


\section{Walks on Simplicial Complexes}\label{sec:highdimwalk}
	Consider a pure $d$-dimensional complex $X$ with a balanced weight function $w:X\to\R_{>0}$. We will call $(X,w)$ a \emph{weighted complex}. For $1\leqslant k< d$, we define a random walk on $X(k)$ known as the \emph{upper $k$-walk} based on movement from an face in $X(k)$ to a higher-dimensional face and then returning to $X(k)$. Similarly, for $1 \leqslant k < d$, we define a random walk on $X(k+1)$ known as the \emph{lower $k$-walk} based on movement from a face in $X(k+1)$ to a lower-dimensional face and then returning to $X(k+1)$. 
	
	To define these walks we construct a bipartite graph $G_k$ with one side corresponding to $X(k)$ and the other side corresponding to $X(k+1)$. We connect $\tau\in X(k)$ to $\sigma\in X(k+1)$ with an edge of weight $w(\sigma)$ iff $\tau\subset \sigma$. Now, consider the simple (weighted) random walk on $G_k$. Given a vertex we choose a neighbor proportional to the weight of the edge connecting the two vertices. 

	This is a walk on a bipartite graph and is naturally periodic. We can consider the odd steps and even steps, in order to obtain two random walks; one on $X(k)$ called $P_{k}^{\wedge}$, and the other on $X(k+1)$ called $P_{k+1}^{\vee}$, where given $\tau\in X(k)$ you take two steps of the walk in $G_k$ to transition to the next $k$-face with respect to the $P_{k}^{\wedge}$ matrix, and similarly, you take two steps in $G_k$ from $\sigma\in X(k+1)$ to transition with respect to $P_{k+1}^{\vee}$.

	Now, let us formally write down the entries of $P_{k}^{\wedge}$ and $P_{k+1}^{\vee}$. Given a simplex $\tau\in X(k)$, first among all $k+1$ dimensional simplices $\sigma\in X(k+1)$ that contain $\tau$ we choose one proportional to $w(\sigma)$. Then, we delete one of the $\card{\sigma}=k+1$ elements of $\sigma$ uniformly at random to obtain a new state $\tau'$. It follows that the probability of transition to $\tau'$ is equal to the probability of choosing $\sigma=\tau\cup\tau'$ in the first step, which is equal to $\frac{w(\tau\cup\tau')}{w(\tau)}$ since $w$ is balanced, times the probability of choosing $\tau'$ conditioned on $\sigma=\tau\cup\tau'$, which is $\frac1{k+1}$. In summary, for $1 \leqslant k < d$,
	\begin{align}\label{eq:P+def}
	    P_{k}^{\wedge}(\tau,\tau') &= \begin{cases}
	        \frac{1}{k+1}, &\quad \text{if } \tau = \tau' \\
	        \frac{w(\tau \cup \tau')}{(k+1)w(\tau)}, &\quad \text{if } \tau \cup \tau' \in X(k+1) \\
	        0, &\quad \text{otherwise}
	    \end{cases} 
	\end{align}
	Note that upper walk is not defined for $k=d$, because there is no $(d+1)$-dimensional simplex in $X$.

	Analogously, given $\sigma\in X(k+1)$, first we remove a uniformly random element of $\sigma$ to obtain $\tau$. Then, among all all $k+1$ simplices $\sigma'\in X(k+1)$ that contain $\tau$ we choose one proportional to $w(\sigma')$. It follows that for $1 \leqslant k < d$,
	\begin{align}\label{eq:P-def}
	    P_{k+1}^{\vee}(\sigma,\sigma') &= \begin{cases}
	        \sum_{\tau \in X(k) : \tau \subset \sigma} \frac{w(\sigma)}{(k+1)w(\tau)}, &\quad \text{if } \sigma = \sigma' \\
	        \frac{w(\sigma')}{(k+1)w(\sigma \cap \sigma')}, &\quad \text{if } \sigma \cap \sigma' \in X(k) \\
	        0, &\quad \text{otherwise}
	    \end{cases}
	\end{align}
	Observe the corresponding random walks are reversible with respect to the weight function $w$, i.e., for all $\tau,\tau'\in X(k)$, we have
	\[w(\tau) P^{\wedge}_k(\tau,\tau')=w(\tau')P^{\wedge}_k(\tau',\tau) \quad\quad w(\tau) P^{\vee}_k(\tau,\tau')=w(\tau')P^{\vee}_k(\tau',\tau).\]
	This implies that both chains have the same stationarity distribution where the probability of $\tau\in X(k)$ is proportional to $w(\tau)$. 

	\begin{lemma}\label{lem:upperlower}
		For any $1\leqslant k<d$, $P_{k}^{\wedge}$ and $P_{k+1}^{\vee}$ are stochastic, self-adjoint w.r.t.\ the $w$-induced inner product, PSD, and have the same (with multiplicity) non-zero eigenvalues.
	\end{lemma}
	\begin{proof}
		Let $P_k$ be the transition probability matrix of the simple random walk on $G_k$. Since $G_k$ is bipartite and we can write 
		\[ P_k=\begin{bmatrix} 0 & P_k^{\downarrow} \\ P_k^{\uparrow} & 0\end{bmatrix}\]
		where $P_k^{\downarrow}\in \R^{X(k+1)\times X(k)}$ and $P_k^{\uparrow}\in \R^{X(k)\times X(k+1)}$ are stochastic matrices. Note that $P_k$ is self-adjoint w.r.t.\ the weight-induced inner product given by weights of the stationary distribution. It follows that $P_k$ is self-adjoint w.r.t. the inner product
		\begin{align*}
		    \langle \phi, \psi \rangle = \sum_{\tau \in X(k)} w(\tau)\phi(\tau)\psi(\tau) + (k+1)\sum_{\sigma \in X(k+1)} w(\sigma)\phi(\sigma)\psi(\sigma)
		\end{align*}
		
		Observe that
		\[ P_k^2 = \begin{bmatrix}	P_k^{\downarrow}P_k^{\uparrow} & 0 \\ 0 & P_k^{\uparrow} P_k^{\downarrow} \end{bmatrix}\]
		So, in particular, $P_k^2$ is PSD and stochastic. Since $P_{k}^{\wedge}$ and $P_{k+1}^{\vee}$ correspond to two step walks on $G_k$, indeed we can write
		\begin{eqnarray*}
			P_{k}^{\wedge} &=& P_k^{\uparrow}P_k^{\downarrow}\\
			P_{k+1}^{\vee} &=& P_k^{\downarrow}P_k^{\uparrow}
		\end{eqnarray*}
		It follows that both matrices are self-adjoint w.r.t.\ the $w$-induced inner product, are PSD, and stochastic, and by \cref{fact:eigenvaluesinvert} they have the same eigenvalues.
	\end{proof}


	Let us specifically study $P_{1}^{\wedge}$. Observe that $P_{1}^{\wedge}$ is the transition probability matrix of  the simple $(1/2)$-lazy random walk  on the weighted 1-skeleton of $X$ where the weight of each edge $e \in X(2)$ is $w(e)$. We also need to consider the non-lazy variant of this random walk, given by the transition matrix 
	\[\tP_{1}^{\wedge} \ = \ 2(P_{1}^{\wedge}-I/2).\] 
	Similarly, for any face $\tau\in X(k)$, we  define the upper  random walk on the faces of the link $X_{\tau}$. Specifically, let $P_{\tau,1}^{\wedge}$ denote the transition matrix of the upper walk, as above, on the 1-dimensional faces of $X_{\tau}$, and 
	\[\tP^{\wedge}_{\tau, 1} \ = \ 2(P_{\tau,1}^{\wedge}-I/2)\]
	be the transition matrix for the non-lazy version. 

	\begin{definition}[Local Spectral Expanders, \cite{KO18}]
		For $\lambda\geqslant 0$, a pure $d$-dimensional weighted complex $(X,w)$ is a $\lambda$-local-spectral-expander if for every $0 \leqslant k < d - 1$, and for every $\tau \in X(k)$, we have $\lambda_2(\tP^{\wedge}_{\tau, 1})\leqslant \lambda$.
	\end{definition}
	In other words, $X$ is $\lambda$-local spectral expander if the spectral gap of the natural random walk on the 1-skeleton of the link of all simplices of $X$ has a spectral gap of at least $1-\lambda$. In this section we give a somewhat simpler proof of the following special case of the main theorem of \cite{KO18}.
	\begin{theorem}[\cite{KO18}]\label{thm:localexpander}
		Let $(X,w)$ be a pure $d$-dimensional weighted $0$-local spectral expander and let $0\leqslant k<d$. Then, for all $-1 \leqslant i\leqslant k$, $P_{k}^{\wedge}$ has at most $\card{X(i)} \leqslant \binom{n}{i}$ eigenvalues of value $> 1-\frac{i+1}{k+1}$, where for convenience, we set $X(-1) = \emptyset$ and $\binom{n}{-1} = 0$. In particular, the second largest eigenvalue of $P_{k}^{\wedge}$ is at most $\frac{k}{k+1}$.
	\end{theorem}
	\begin{remark}
	In other words, $P_{k}^{\wedge}$ has very few ``big'' eigenvalues. For example, $P_{k}^{\wedge}$ has exactly one eigenvalue strictly larger than $\frac{k}{k+1}$ corresponding to the maximum eigenvalue (which has value 1) and at most $n = \card{X(1)}$ eigenvalues strictly larger than $\frac{k-1}{k+1}$. Hence, $P_{k}^{\wedge}$ has at most $n -1$ eigenvalues between $\frac{k-1}{k+1}$ and $\frac{k}{k+1}$. Note that the significance of this theorem is that we are able to establish an estimate on \emph{all} eigenvalues of $P_{k}^{\wedge}$.
	\end{remark}
For the proof, we will need the following lemma. We remark that the inner product on the space $\R^{X(k)}$ is given by $\dotprod{\phi,\psi}=\sum_{\tau\in X(k)} w(\tau)\phi(\tau)\psi(\tau)$, and that being self-adjoint, PSD, and the Loewner order are defined w.r.t.\ this inner product.
\begin{lemma}\label{lem:lowerupper}
$P_{k}^{\wedge} \preccurlyeq \frac{k}{k+1} P_{k}^{\vee} + \frac{1}{k+1}I$ for all $0 \leqslant k < d$.
\end{lemma}
\begin{proof}
For convenience, let $M = P_{k}^{\wedge} - \parens*{\frac{k}{k+1} P_{k}^{\vee} + \frac{1}{k+1}I}$. Fix $\eta \in X(k-1)$. We will first consider submatrices $M_{\eta}$ whose entries are given by the following:
\begin{align*}
    M_{\eta}(\tau,\sigma) = \begin{cases}
        M(\tau,\sigma), & \quad\text{if } \tau \neq \sigma, \eta = \tau \cap \sigma \\
        -\frac{1}{k+1} \cdot \frac{w(\tau)}{w(\eta)}, &\quad\text{if } \tau = \sigma, \tau \supset \eta \\
        0, &\quad\text{otherwise}
    \end{cases}
\end{align*}
Note that $M = \sum_{\eta \in X(k-1)} M_{\eta}$ and hence, it suffices to prove that $M_{\eta} \preccurlyeq 0$ for every $\eta \in X(k-1)$. \\~\\
Fix $\eta \in X(k-1)$. Let $\tau,\sigma \in X(k)$ with $\tau \neq \sigma$ and $\tau \cap \sigma = \eta$. Then
\begin{align*}
    M_{\tau \cap \sigma}(\tau,\sigma) = M(\tau, \sigma) = \frac{1}{k+1}\parens*{\frac{w(\tau \cup \sigma)}{w(\tau)} - \frac{w(\sigma)}{w(\tau \cap \sigma)}} = \frac{1}{k+1} \parens*{\frac{w(\tau\cup\sigma)w(\tau \cap \sigma) - w(\tau)w(\sigma)}{w(\tau)w(\tau \cap \sigma)}}
\end{align*}
Furthermore, by definition, if $\tau \in X(k)$ with $\tau \supset \eta$, then $M_{\eta}(\tau,\tau) = -\frac{1}{k+1} \cdot \frac{w(\tau)}{w(\eta)}$. A matrix calculation reveals that
\begin{align*}
    M_{\eta} = \frac{1}{(k+1)w(\eta)}\diag(w_{\eta})^{-1} \cdot (w(\eta) \cdot A_{\eta} - w_{\eta}w_{\eta}^{\intercal})
\end{align*}
where $w_{\eta}$ is the $\card{X(k)}$-dimensional vector whose non-zero entries are $w(\tau)$ for $\tau \supset \eta$, and $A_{\eta}$ is the $\card{X(k)} \times \card{X(k)}$ matrix whose non-zero entries are $w(\tau \cup \sigma)$ for $\tau,\sigma \in X(k)$ satisfying $\tau \cup \sigma \in X(k+1)$ and $\tau \cap \sigma = \eta$. Note that $M_\eta$ is NSD w.r.t.\ the inner product defined by $w$, if and only if $\diag(w_\eta)M_\eta$ is NSD in the usual sense, because for any $v$
\[ \dotprod{v, M_\eta v}=v^\intercal \diag(w_k) M_\eta v=v^\intercal \diag(w_\eta)M_\eta v, \]
where $w_k$ is the vector of $w$ values on $X(k)$ and for the last equality we used that $w_k$ is the same as $w_\eta$ on all $\tau\supset \eta$.

Thus, it suffices to prove that $A_{\eta} \preccurlyeq \frac{w_{\eta}w_{\eta}^{\intercal}}{w(\eta)}$. We view $A_{\eta}$ as the weighted adjacency matrix of the 1-skeleton (which we recall is a graph) of the link $X_{\eta}$. Then $\tP_{\eta,1}^{\wedge} = \frac{1}{k+1}\diag(w_{\eta})^{-1}A_{\eta}$ gives its non-lazy simple random walk matrix. As $(X,w)$ is a 0-local spectral expander, $\tP_{\eta,1}^{\wedge}$ has at most one positive eigenvalue, whence $A_{\eta} = (k+1)\diag(w_{\eta}) \cdot \tP_{\eta,1}^{\wedge}$ has at most one positive eigenvalue by \cref{lem:stochastic}.

Finally, observe that the weights being balanced enforces that $w(\tau) = \sum_{\sigma \in X(k) : \tau \cup \sigma \in X(k+1)} w(\tau \cup \sigma)$ and $w(\eta) = \sum_{\tau \in X(k) : \tau \supset \eta} w(\tau)$. That $A_{\eta} \preccurlyeq \frac{w_{\eta}w_{\eta}^{\intercal}}{w(\eta)}$ then follows immediately by \cref{lem:adjacency1poseig}.
\end{proof}
\begin{proof}[Proof of \cref{thm:localexpander}]
We go by induction on $k$. The case $k=0$ is trivial, as $P_{0}^{\wedge}$ is $1 \times 1$. When $k = 1$, we have $P_{1}^{\wedge} = \frac{1}{2}\parens*{\tP_{1}^{\wedge} + I}$. As $(X,w)$ is a 0-local spectral expander, $\tP_{1}^{\wedge}$ has exactly one positive eigenvalue, with value 1. Hence, $P_{1}^{\wedge}$ has eigenvalue 1 with multiplicity 1. All other eigenvalues of $P_{1}^{\wedge}$ are less than or equal to $1/2$, of which, there are $\card{X(1)} - 1$ many. Thus, the base case holds.

Assume the claim holds for some $d - 1 > k \geqslant 0$. Recall by \cref{lem:upperlower}, $P_{k+1}^{\vee}$ has the same non-zero eigenvalues as $P_{k}^{\wedge}$. By \cref{lem:lowerupper},
\begin{align*}
    P_{k+1}^{\wedge} \preccurlyeq \frac{k+1}{k+2} P_{k+1}^{\vee} + \frac{1}{k+2}I
\end{align*}
For $-1 \leqslant i \leqslant k$, $P_{k}^{\vee}$ as at most $\card{X(i)}$ eigenvalues $> 1 - \frac{i+1}{k+1}$. Hence, $P_{k+1}^{\wedge}$ has at most $\card{X(i)}$ eigenvalues $> \frac{k+1}{k+2} \cdot \parens*{1 - \frac{i+1}{k+1}} + \frac{1}{k+2} = 1 - \frac{i+1}{k+2}$. For $i = k+1$, we trivially have that $P_{k+1}^{\wedge}$ has at most $\card{X(k+1)}$ eigenvalues $> 0$, as $P_{k+1}^{\wedge}$ is $\card{X(k+1)} \times \card{X(k+1)}$.
\end{proof}

\section{From Strongly Log-Concave Polynomials to Local Spectral Expanders}\label{sec:SLCtoLSE}
In this section, we prove \cref{thm:SLCmixing}. 
Let $p = \sum_Sc_Sx^S \in \R[x_1,\dots,x_n]$ be a $d$-homogeneous multiaffine strongly log-concave polynomial with nonnegative coefficients. We can construct a pure $d$-dimensional complex $X^p$ from $p$ as follows: For every term $c_S x^S$ of $p$ we include the $d$-dimensional simplex 
$S$ with weight
\begin{equation}\label{eq:wX^p} 
w(S)=c_S.	
\end{equation}
Note that since $p$ has nonnegative coefficients, the above weight function is nonnegative. 
We turn $X^p$ into a simplicial complex by including all subsets of the $d$-dimensional simplices and weighting each lower dimensional simplex inductively according to \cref{eq:balancedweight}.

\begin{proposition}\label{prop:SLCtoHDE}
Let $p\in \R[x_1,\dots,x_n]$ be a multiaffine homogeneous polynomial with nonnegative coefficients. If $p$ is strongly log-concave then $(X^p,w)$ is a $0$-local-spectral-expander, where $w(S) = c_{S}$ for every maximal face $S \in X^{p}$.
\end{proposition}
The converse of the above statement also holds true and we will discuss it in the next section.
We build up to the proof of \cref{prop:SLCtoHDE} and thereby the proof of \cref{thm:SLCmixing}.

We now fix a multiaffine $d$-homogeneous strongly log-concave polynomial $p$ with nonnegative coefficients. 
Fix a simplex $\tau\in X^p(k)$ and let $p_\tau=\left(\prod_{i\in\tau} \partial_i\right) p$.
Note that $p_\tau$ is $(d-k)$-homogeneous.

\begin{lemma}\label{lem:wTau_pTau}
	For any $0\leqslant k\leqslant d$, and any simplex $\tau\in X^p(k)$, 
	$ w(\tau)=(d-k)!\cdot p_{\tau}(\bone)$.
\end{lemma}
\begin{proof}
	We prove this by induction on $d-k$. If $\dim(\tau)=d$ then $p_{\tau}=c_{\tau}$ and the statement follows immediately from \cref{eq:wX^p}.
	So, suppose the statement holds for all simplices $\sigma \in X^p(k+1)$ and fix a simplex $\tau\in X^p(k)$.
 Then by definition, 
\[w(\tau) \ = \ \sum_{\sigma\in X^p(k+1): \sigma\supset \tau} w(\sigma) \ = \ \sum_{i\in X^p_\tau(1)} w(\tau\cup i).\]
Using the inductive hypothesis and the fact that  $\partial_i p_\tau=0$ for $i\notin X^p_\tau(1)$, we then find that 
\[	w(\tau) 
	\ = \ (d-k-1)!\sum_{i\in X^p_\tau(1)} p_{\tau \cup \set{i}}(\bone)
	\ = \ (d-k-1)!\sum_{i=1}^n \partial_i p_\tau(\bone) 
	\ = \ (d-k)!\cdot p_{\tau}(\bone),
\]
where the last equality follows from Euler's identity. 
\end{proof}

Recall that $\tP^{\wedge}_{\tau,1}$ is the transition probability matrix of the non-lazy random walk on the 1-skeleton of the link $X^p_\tau$. 
To prove the \cref{prop:SLCtoHDE} it is enough to show that $\lambda_2(\tP^{\wedge}_{\tau,1})\leqslant 0$, i.e. that $\tP^{\wedge}_{\tau,1}$ has at most one positive eigenvalue.

\begin{proof}[Proof of \cref{prop:SLCtoHDE}] 
Since $p$ is strongly log-concave, $\Hess{p_\tau}(\bone)$ has at most one positive eigenvalue. Let 
\begin{equation}\label{eq:tHessian}
\tilde{\nabla}^{2} p_\tau = \frac{1}{d-k-1}\diag(\nabla p_\tau(\bone))^{-1} \Hess{p_\tau(\bone)}.
\end{equation}
We claim that
\begin{equation} 
\tilde{\nabla}^2 p_\tau=\tP^{\wedge}_{\tau,1}. 
\label{eq:tnablatP}
\end{equation}
To see this, note that by \cref{eq:P+def}, for $i,j\in X_{\tau}^p$,  
\[\tP^{\wedge}_{\tau,1}(i,j) =\frac{w_{\tau}(\set{i,j})}{w_\tau(\set{i})}=\frac{w(\tau\cup \set{i,j})}{w(\tau\cup \set{i})}.\]
On the other hand, by \cref{eq:tHessian},
\[ (\tilde{\nabla}^2 p_\tau) (i,j) = \frac{(\partial_i\partial_j p_\tau)(\bone)}{(d-k-1)\cdot (\partial_i p_\tau)(\bone)}.  \]
The above two are equal by \cref{lem:wTau_pTau}, which proves \cref{eq:tnablatP}.

Since $p$ has nonnegative coefficients, the vector $\nabla p_\tau(\bone)$ has nonnegative entries, which implies $\diag(\nabla p_\tau(\bone))\succcurlyeq 0$. Since $\nabla^2 p_\tau(\bone)$ has at most one positive eigenvalue, $\tilde{\nabla}^2p_\tau(\bone)$ has at most one positive eigenvalue by \cref{lem:stochastic} as desired. Therefore by \cref{eq:tnablatP}, $\tilde{\nabla}^2p_\tau=\tP^{\wedge}_{\tau,1}$ has at most one positive eigenvalue and $\lambda_2(\tP^{\wedge}_{\tau,1})\leqslant 0$.
\end{proof}
	
With this we are ready to prove our main theorem.

\begin{proof}[Proof of \cref{thm:SLCmixing}]
Let $\mu$ be a $d$-homogeneous strongly log-concave distribution and let $P_{\mu}$ be the transition probability matrix of the chain $\cM_\mu$. By \cref{thm:mixingtime} it is enough to show that $\lambda^*(P_{\mu})\leqslant 1-1/d$. Observe that the chain $\cM_\mu$ is exactly the same as the chain $P^{\vee}_{d}$ for the simplicial complex $X^{g_\mu}$ defined above. Therefore, $\lambda^*(P_{\mu})=\lambda^*(P^{\vee}_{d}) = \lambda^{*}(P_{d-1}^{\wedge})$, where the last equality follows by \cref{lem:upperlower}. Since $g_\mu$ is strongly log-concave, by \cref{prop:SLCtoHDE}, $X^{g_\mu}$ is $0$-local-spectra-expander. Therefore, by \cref{thm:localexpander},
\[ \lambda^{*}(P_{d-1}^{\wedge}) \leqslant 1-\frac{1}{(d-1)+1}=1-\frac1d,
\]
as desired.
\end{proof}

\subsection{Bases Exchange Walk}
In this part we prove \cref{thm:basesexchange}.
Fix a rank $r$ matroid $M=([n],\cI)$, let $\mu$ denote the uniform distribution on the bases of $M$, 
and consider the simplicial complex $X^{g_{\mu}}$. As before, recall that $P_{r}^{\vee}$ is the transition matrix for the Markov chain  $\cM_{\mu}$ defined in the introduction. As discussed in \cref{subsec:randomwalks}, each reversible Markov chain is equivalent to a random walk in a (weighted) undirected graph. Let $H_M$ be the graph corresponding to $P_{r}^{\vee}$. Then the vertices of $H_M$ correspond to bases of $M$, and the weight of an edge between two bases $\tau,\tau'$ is, 
\[P_{r}^{\vee}(\tau,\tau')=P_{r}^{\vee}(\tau',\tau).\]
The symmetry of $P_{r}^{\vee}$ follows by the fact that $P_{r}^{\vee}$ is reversible and $w(\tau)=w(\tau')=1$.
Observe that $P_{r}^{\vee}$ is the adjacency matrix of $H_M$. Furthermore, note that by \cref{thm:gMclc} (proof in \cref{sec:slc}), $g_{\mu}$ is a strongly log-concave polynomial and hence, $X^{g_{\mu}}$ is a 0-local spectral expander by \cref{prop:SLCtoHDE}. By \cref{thm:localexpander}, $\lambda_2(P_{r}^{\vee}) \leqslant 1-\frac1{r}$. 
Note that the weighted degree of each vertex of $H_M$ is 1 since $P_{r}^{\vee}$ is a stochastic matrix.
Therefore, by Cheeger's inequality, \cref{thm:Cheeger}, 
\begin{equation}\label{eq:condHM}\cond(H_M) \ \geqslant \ \frac{1-\lambda_2(P_{r}^{\vee})}{2} \ \geqslant  \ \frac{1-(1-1/r)}{2}=\frac1{2r}.	
\end{equation}

Let $G_M=(\cB,E)$ be the (unweighted) bases exchange graph associated to $M$ as defined in the introduction. It follows that $G_M$ is the unweighted base graph of $H_M$. Fix a nonempty set $S\subset \cB$ of bases such that $\card{S}\leqslant \card{\cB}/2$.
We need to show that the expansion $\expansion(S)$ is $\geqslant 1$.
Note that since the weighted degree of each vertex of $H_M$ is 1, $\vol_H(S)=\card{S}$. It follows that
\[ \cond(H_M)\leqslant \cond(S) = \frac{\sum_{\tau,\tau':\tau \in S,\tau'\notin S} P_{r}^{\vee}(\tau,\tau')}{\card{S}} \leqslant \frac{\sum_{\tau,\tau':\tau \in S,\tau'\notin S} \frac1{2r}}{\card{S}} = \frac{\frac1{2r} \card{E(S,\overline{S})}}{\card{S}}=\frac{\expansion(S)}{2r}
\]
The second inequality follows by the following fact. Putting the above together with \cref{eq:condHM}, implies $\expansion(S)\geqslant 1$ as desired. This completes the proof of \cref{thm:basesexchange}.
\begin{fact} 
For any pair of bases, $\tau,\tau'\in X^{g_{\mu}}(r)$, 
\[P^{\vee}_{r}(\tau,\tau')\leqslant \frac1{2r}.\]
\end{fact}
\begin{proof}
Suppose $P_{r}^{\vee}(\tau,\tau')>0$. This means that $\card{\tau\cap\tau'}=r-1$. Therefore, 
\[w(\tau\cap \tau')\geqslant w(\tau)+w(\tau') =2.\]
So, $P_{r}^{\vee}(\tau,\tau')=\frac{w(\tau')}{rw(\tau\cap\tau')}\leqslant \frac1{2r}.$
\end{proof}

\section{Proof of Strong Log-Concavity and Applications}\label{sec:slc}
\label{sec:applications}
In this part we prove \cref{thm:SLCdeg2} using connections with high dimensional expanders. 
\cref{thm:SLCdeg2} can be seen as the converse of \cref{prop:SLCtoHDE}.
In fact, the following statement was proved by \textcite{Opp18}.
\begin{theorem}[\cite{Opp18}]\label{thm:opp} Let $(X,w)$ be a pure $d$-dimensional weighted simplicial complex such that:
\begin{enumerate}
    \item for all $\tau\in X(k)$, for $0\leqslant k\leqslant d-2$, the 1-skeleton of $X_\tau$ is a connected (weighted) graph, and 
    \item for any $\tau\in X(d-2)$, $\lambda_2(\tP^{\wedge}_{\tau,1})\leqslant 0$.
\end{enumerate}
Then $(X,w)$ is a weighted $0$-local-spectral-expander.
\end{theorem}
Our proof of \cref{thm:SLCdeg2} can be seen as a translation of Oppenheim's result into the language of polynomials.\footnote{In a companion paper, we will give an alternative proof based purely on elementary calculus and linear algebra.} We actually prove a slightly stronger statement, in the sense that a direct translation of \cref{thm:opp} corresponds to testing strong log-concavity of a \emph{multiaffine} homogeneous polynomial. In this case, each derivative $\partial_{\tau}p$ corresponds to taking the link of $\tau$ in $X^{p}$. Indecomposability of each derivative then corresponds to connectivity of the 1-skeleton of the corresponding link, and log-concavity of quadratics corresponds to $\lambda_{2}(\tP_{\tau,1}^{\wedge})\leqslant 0$.

\begin{proof}[Proof of \cref{thm:SLCdeg2}]
We proceed by induction on the degree of $p$. If the degree of $p$ is at most 2, the claim obviously holds. So, suppose $d\geqslant 3$. For any $1\leqslant i\leqslant n$, let $p_i$ denote $\partial_i p$.  By induction, we can assume that for all $i$, $p_i$ is strongly log-concave (at $\bone$).
%
%

First, by \cref{eq:Hesssum}, $\Hess p(\bone)=\frac1{d-2}\sum_{i=1}^n \Hess p_i (\bone)$. By induction and \cref{prop:oneposeigenvalue}, each matrix $\Hess p_i(\bone)$ has at most one positive eigenvalue. 

Instead, we work with the normalized Hessian matrix, $\tilde{\nabla}^2 p=\frac{1}{d-1} \diag(\nabla p(\bone))^{-1} \Hess p(\bone)$ as defined in \cref{eq:tHessian}.
 Since the normalized Hessian matrix is stochastic, its top eigenvector  is the all-ones vector. When working with the normalized Hessian we need to use the correct inner product operators. For a $d$-homogeneous polynomial $p$ with nonnegative coefficients and $d > 1$, and vectors $\phi,\psi\in\R^n$, define
\begin{align*}
    \dotprod{\phi,\psi}_{p} = (d-1) \sum_{j=1}^{n} \phi(j)\psi(j)   (\partial_j p(\bone)),
\end{align*}
which gives the norm $\norm{\phi}_p^2=\dotprod{\phi,\phi}_p$.
The following identity is immediate:
    \begin{align*}
        \dotprod{\phi, (\tilde{\nabla}^{2}p)\psi}_{p} = \dotprod{\phi, \Hess p(\bone) \psi} = \dotprod{(\tilde{\nabla}^{2}p)\phi, \psi}_{p}
    \end{align*}
    In particular, $\tilde{\nabla}^2 p$ is self-adjoint with respect to $\dotprod{\cdot ,\cdot}_{p}$. Furthermore, by \cref{eq:Hesssum}, 
  \begin{equation}\label{eq:tHessumi}\dotprod{\phi,(\tilde{\nabla}^2 p)\psi}_{p}
  = \dotprod{\phi, \Hess p(\bone)\psi}=
  \frac1{d-2}\sum_{k=1}^n \dotprod{\phi, \nabla^2 p_k(\bone) \psi}=\frac1{d-2}\sum_{k=1}^n \dotprod{\phi,\tilde{\nabla}^2 p_k \psi}_{p_k}.
  \end{equation}
We highlight that the Hessian $\nabla^{2}p(\bone)$ may be viewed as the weighted adjacency matrix of a graph with edge weights $\partial_{i}\partial_{j}p(\bone)$. Thus, our normalized Hessian may be viewed as the associated random walk matrix, and the inner product $\langle\cdot,\cdot\rangle_{p}$ may be viewed as a change of basis, which converts the random walk matrix into the normalized adjacency matrix.

Let $\mu$ be an eigenvalue of $\tilde{\nabla}^{2}p$ with eigenvector $\phi$. We prove that $\mu\leqslant \mu^{2}$. We claim that this is enough for the induction step: First, since $\tilde{\nabla}^{2}p$ is stochastic $\mu\leqslant 1$. Therefore, we either $\mu=1$ or $\mu\leqslant 0$. So, to prove that $\tilde{\nabla}^2 p$ has (exactly) one positive eigenvalue, it is enough to show that $\lambda_2(\tilde{\nabla}^{2}p)<1$. But, since $p$ is indecomposable,  the underlying (weighted) graph of $\tilde{\nabla}^2 p$ is connected, so by  \cref{cor:cheegerconnected}, $\lambda_2(\tilde{\nabla}^2 p)<1$.

It remains to prove that $\mu\leqslant \mu^2$.
From \cref{eq:tHessumi},
\begin{align}\label{eq:muphitnabla}
    \mu \norm{\phi}_{p}^{2} &= \dotprod{ \phi, (\tilde{\nabla}^{2} p)\phi }_{p} = \frac{1}{d-2} \sum_{k=1}^{n} \dotprod{ \phi, (\tilde{\nabla}^{2} p_{k})\phi }_{p_k}
\end{align}
Decomposing $\phi$ orthogonally along $\bone$ write 
\[\phi=\phi_{k}^{\perp\bone} + \phi_{k}^{\bone}\] 
where $\phi_k^{\bone}=\frac{\dotprod{ \phi,\bone}_{p_k}}{\norm{\bone}_{p_k}^2}\bone$ and $\phi_k^{\perp \bone}$ is orthogonal to $\bone$, i.e., $\dotprod{ \phi_k^{\perp\bone},\bone}_{p_k}=0$.
It follows that \[\dotprod{ \phi_k^{\perp\bone}, (\tilde{\nabla}^2 p_k) \phi_k^{\perp\bone}}_{p_k}\leqslant 0.\] This is because  $\tilde{\nabla}^{2}p_{k}$ has exactly one positive eigenvalue with corresponding eigenvector of $\bone$. Therefore,
\begin{align}
    \mu\norm{\phi}^2_p\leqslant \frac{1}{d-2} \sum_{k=1}^{n} \dotprod{ \phi_{k}^{\bone}, (\tilde{\nabla}^{2}p_{k})\phi_{k}^{\bone} }_{p_k} = \frac{1}{d-2} \sum_{k=1}^{n} \dotprod{ \phi_{k}^{\bone}, \phi_{k}^{\bone} }_{p_k} = \frac{1}{d-2} \sum_{k=1}^{n} \frac{\dotprod{ \phi, \bone }_{p_k}^{2}}{\dotprod{ \bone,\bone }_{p_k}}\label{eq:muphitnabla2}
\end{align}
Next, we rewrite the numerator and denominator of each ratio in the righthand side. We have
\begin{align*}
    \dotprod{ \bone,\bone }_{p_k} = (d-2)\sum_{i=1}^{n}  (\partial_{i}p_{k}(\bone)) = (d-2)(d-1) \cdot p_{k}(\bone)
\end{align*}
where we used \cref{eq:eulerppartial} for polynomial $p_k$.
Furthermore,
\begin{align}\label{eq:phionepk}
    \dotprod{ \phi, \bone }_{p_k} = (d-2)\sum_{i=1}^{n} \phi(i)  (\partial_{i}p_{k}(\bone)) = (d-2) \cdot ((\nabla^{2}p(\bone))\phi)(k)
\end{align}
So, putting the above identities together, we obtain
\begin{align*}
    \frac{\dotprod{ \phi, \bone }_{p_k}}{\dotprod{ \bone,\bone }_{p_k}} = \frac{1}{(d-1) \cdot p_{k}} ((\nabla^{2}p(\bone)) \phi)(k) = ((\tilde{\nabla}^{2}p(\bone))\phi)(k) = \mu \cdot \phi(k)
\end{align*}
where the last identity we crucially used $\phi$ is the eigenvector of $\tilde{\nabla}^2p$ corresponding to $\mu$. Plugging into \cref{eq:muphitnabla2} we get
\begin{align*}
    \mu\norm{\phi}_p^2&\leqslant \frac{1}{d-2} \sum_{k=1}^{n} \frac{\dotprod{ \phi, \bone }_{p_k}^{2}}{\dotprod{ \bone,\bone }_{p_k}} 
    =\frac{\mu}{d-2}\sum_{k=1}^n \phi(k) \dotprod{ \phi,\bone}_{p_k}\\
    &=\mu\sum_{k=1}^n \phi(k) ((\Hess p(\bone))\phi)(k) = \mu \dotprod{ \phi, (\Hess p(\bone))\phi}=\mu\dotprod{ \phi,(\tilde{\nabla}^2 p)\phi}_p = \mu^2 \norm{\phi}_p^2.
\end{align*}
The second equality uses \cref{eq:phionepk} and the last equality uses definition of $\phi$. So, $\mu\leqslant \mu^2$ as desired.
\end{proof}

\subsection{An Alternative Proof of Strong Log-concavity of  Bases Generating Polynomial}
\begin{theorem}\label{thm:gMclc}
Let $M = ([n],\cI)$ be a matroid of rank $r$. Then, for every choice of ``external field'' $\bm{\lambda} = (\lambda_{1},\dots,\lambda_{n}) \in \R_{>0}^{n}$, its (weighted) bases generating polynomial
\begin{align*}
    g_{M}(x_{1},\dots,x_{n}) = \sum_{B \text{ basis}} \bm{\lambda}^{B}x^{B}
\end{align*}
is strongly log-concave (at $\bone$). 
\end{theorem}
\begin{proof} 
We verify the indecomposability and log-concavity conditions of \cref{thm:SLCdeg2}. If $i_{1},\dots,i_{k}$ contains duplicate elements, then multiaffine-ness of $g_M$ forces $\partial_{i_{1}}\dotsb \partial_{i_{k}}g_M$ to be identically zero. 
For any subset $S = \set{i_{1},\dots,i_{k}}$, we use $\partial_{S}$ as a shorthand notation for $\partial_{i_{1}}\dotsb \partial_{i_{k}}$. 
If $S$ is not independent, then no basis contains $S$, and again, $\partial_{S}g_M = 0$ identically. Hence, we assume $S \in \cI$. 

We first argue that $\partial_{S}g_{M}$ is indecomposable. Observe that $\partial_{S}g_{M}$ equals the weighted basis generating polynomial 
$ g_{M/S}$ of the contraction $M/S$. As $M/S$ is a matroid of rank $\geqslant 2$, applying the exchange property immediately tells us that $g_{M/S}$ is indecomposable. 

Now, we verify log-concavity of all quadratics. Assume $S \in \cI$ and $\card{S} = r-2$. As $S$ has rank-$(r-2)$, $M/S$ has rank two. In particular, $\partial_{S}g_{M} = g_{M/S}$ is quadratic and $\nabla^{2}g_{M/S}$ has entries
\begin{align*}
    (\nabla^{2}g_{M/S})_{ij} = \begin{cases}
        \lambda_{i}\lambda_{j}, &\quad \text{if } \set{i,j} \text{ is independent in } M/S \\
        0, &\quad \text{otherwise}.
    \end{cases}
\end{align*}
We need to prove that the above matrix has at most one positive eigenvalue.

For any set $T\subseteq [n]$, let $\bm{\lambda}_T$ denote the vector with $i$th entry $\lambda_i$ for $i\in T$ and $0$ for $i\not\in T$. 
The matroid partition property tells us that we can partition the non-loops of $M/S$
into blocks $B = B_{1} \cup \dots \cup B_{k}$. 
Then
\begin{align*}
    \nabla^{2}g_{M/S}  \ = \  \bm{\lambda}_{B}\bm{\lambda}_{B}^{\intercal} - \sum_{i=1}^{k} \bm{\lambda}_{B_{i}} \bm{\lambda}_{B_{i}}^{\intercal} \preccurlyeq \bm{\lambda}_{B}\bm{\lambda}_{B}^{\intercal} .
\end{align*}
This proves that $\nabla^{2} \partial_{S}g_{M}$ has at most one positive eigenvalue and thus $\partial_{S}g_{M}$ is log-concave.
\end{proof}
\begin{corollary}
Let $M = ([n],\cI)$ be a matroid of rank $r$. Then the polynomial
\begin{align*}
\sum_{S \in \cI : \card{S} = k} \bm{\lambda}^{S}x^{S}
\end{align*}
is strongly log-concave (at $\bone$), for every choice of ``external field'' $\bm{\lambda} = (\lambda_{1},\dots,\lambda_{n}) \in \R_{>0}^{n}$.
\end{corollary}
\begin{proof}
This is the weighted basis generating polynomial of the rank-$k$ truncation of $M$, which is still a matroid. Thus, the claim follows from \cref{thm:gMclc}.
\end{proof}
\subsection{The Random Cluster Model on Matroids}
Using the the simplified characterization of strong log-concavity, we can also prove \cref{thm:randomclusterSLC}.

\begin{proof}[Proof of  \cref{thm:randomclusterSLC}]
We verify the indecomposability and log-concavity conditions of \cref{thm:SLCdeg2}. Let $f$ denote $f_{M,k,q}$. 
As before, for the derivatives $\partial_{i_1}\hdots \partial_{i_k}f$, by multiaffine-ness of $f$, 
we may assume $i_{1},\dots,i_{k}$ does not contain any duplicate elements. 
 Consider $S = \set{i_{1},\dots,i_{k}}$. 
Note that $\partial_{S}f$ has a monomial $x^{T}$ with nonzero coefficient for every $T \subset [n] \setminus S$ with $\card{T} \leqslant k - \card{S}$. Hence, indecomposability of $\partial_{S}f$ for every $S \subset [n]$ with $\card{S} \leqslant k$ is immediate. 

Now, we verify log-concavity of quadratics. Let $S \subset [n]$ with $\card{S} = k-2$. First, we calculate that
\[
    (\partial_{S}f)(x_{1},\dots,x_{n}) 
    = \sum_{T \in \binom{[n]}{k} : T \supset S} q^{-\rank(T)}\bm{\lambda}^{T \setminus S}x^{T \setminus S} 
    = \sum_{\set{i,j} \in \binom{[n] \setminus S}{2}} q^{-\rank(S \cup \set{i,j})} \lambda_{i}\lambda_{j}x_{i}x_{j}.
\]
Then for elements $i\neq j$ of $[n]\setminus S$, the $(i,j)$th entry of $\nabla^{2}\partial_{S}f$ is 
\begin{align*}
    (\nabla^{2}\partial_{S}f)_{ij}  = q^{-\rank(S \cup \set{i,j})}\lambda_{i}\lambda_{j} = q^{-\rank(S)}q^{-\rank_{M/S}(\set{i,j})}\lambda_{i}\lambda_{j}
\end{align*}
We will show that the matrix $A = q^{\rank(S)}\nabla^{2}\partial_{S}f$ at most one positive eigenvalue. 
Note that for $i\neq j$ in $[n]\setminus S$, $A_{ij} = q^{-\rank_{M/S}(\set{i,j})}\lambda_{i}\lambda_{j}$.
From here, consider the vector $v\in \R^n$ with $v_i = 0$ for $i\in S$, $v_i = \lambda_{i}$ for loops of $M/S$ and $v_i = q^{-1}\lambda_{i}$ 
for non-loops of $M/S$. 

Now consider the matrix $vv^{\intercal} - A$, we can check that for $i,j\in [n]\setminus S$, 
\begin{align*}
    (vv^{\intercal} - A)_{ij} = \begin{cases}
        (q^{-2} - q^{-1})\lambda_{i}\lambda_{j}&\quad\text{if } i,j \text{ are parallel non-loops in $M/S$, and}  \\
        0&\quad\text{otherwise}
    \end{cases}
\end{align*}
In particular, by the matroid partition property, if $B_{1},\dots,B_{k}$ denote the equivalence classes of non-loops of $M/S$ which are parallel to each other, then
\begin{align*}
    vv^{\intercal} - A = (q^{-2} - q^{-1})\sum_{j=1}^{k} \bm{\lambda}_{B_{j}}\bm{\lambda}_{B_{j}}^{\intercal}
\end{align*}
where $\bm{\lambda}_{B_{j}}$ is the vector with entries $\bm{\lambda}_{B_{j}}(i) = \lambda_{i}$ if $i \in B_{j}$ and $\bm{\lambda}_{B_{j}}(i) = 0$ otherwise.

As $0 < q \leqslant 1$, $q^{-2} - q^{-1} \geqslant 0$, in which case the right-hand side is positive semidefinite and $A \preccurlyeq vv^{\intercal}$. We conclude $A$ has at most one positive eigenvalue as desired.
\end{proof}

\begin{remark}
Observe that
\begin{align*}
    q^{r}f_{M,r,q}(x_{1},\dots,x_{n})
\end{align*}
converges to the bases generating polynomial coefficient-wise as $q \rightarrow 0$. Hence, one can view this as a stronger result than strong log-concavity of the bases generating polynomial.
\end{remark}
\subsection{Geometric Scaling of Coefficients}
In this section we prove \cref{thm:cpow}.

\begin{proof}[Proof of \cref{thm:cpow}]
If $k = 0,1$, the claim is obvious so assume $k \geqslant 2$. The claim is obvious when $\alpha = 1$, and the case $\alpha = 0$ follows by taking coefficient-wise limits as $\alpha \rightarrow 0$. Hence, we will also assume $0 < \alpha < 1$. Finally, we will assume that all coefficients $c_{S}$ are strictly positive. The result for general strongly log-concave polynomials then follows by taking coefficient-wise limits.  \\~\\
Let $T \in \binom{[n]}{k-2}$. We must prove that $\nabla^{2} \partial_{T}f_{\alpha}$ has at most one positive eigenvalue. Observe that we may concisely write
\begin{align*}
    \nabla^{2} \partial_{T}f &= \brackets*{c_{T \cup \set{i,j}}}_{ij} \\
    \nabla^{2} \partial_{T}f_{\alpha} &= \brackets*{c_{T \cup \set{i,j}}^{\alpha}}_{ij}
\end{align*}
As $\nabla^{2}\partial_{T}f$ has at most one positive eigenvalue, and all entries are nonnegative, we may write
\begin{align*}
    \nabla^{2} \partial_{T}f = vv^{\intercal} - A
\end{align*}
for some vector $v \in \R^{n}$ and a positive semidefinite matrix $A$. Note that since $\nabla^{2}\partial_{T}f$ has strictly positive entries, the Perron-Frobenius Theorem (see \cref{thm:perron}) tells us that the entries of $v$ are strictly positive. In particular, $c_{T \cup \set{i,j}} = v_{i}v_{j} - A_{ij} > 0$ where $v_{i}v_{j} > 0$. Our goal is to write
\begin{align*}
    c_{T\cup\set{i,j}}^{\alpha} = (v_{i}v_{j} - A_{ij})^{\alpha} = v_{i}^{\alpha}v_{j}^{\alpha} \parens*{1 - \frac{A_{ij}}{v_{i}v_{j}}}^{\alpha}
\end{align*}
and then Taylor expand $\parens*{1 - \frac{A_{ij}}{v_{i}v_{j}}}^{\alpha}$. Consider the function $\varphi_{\alpha}(x) = (1 - x)^{\alpha}$, whose Taylor expansion about zero we recall is
\begin{align*}
    \sum_{k=0}^{\infty} \frac{\prod_{j=0}^{k-1} (\alpha - j)}{k!} \cdot (-1)^{k}x^{k} = \sum_{k=0}^{\infty} \parens*{\prod_{j=0}^{k-1} \frac{\alpha - j}{1 + j}} \cdot (-1)^{k}x^{k} = 1 - \sum_{k=1}^{\infty} \parens*{\prod_{j=0}^{k-1} \abs*{\frac{\alpha - j}{1 + j}}} \cdot x^{k}
\end{align*}
where for the last equality, we crucially use the fact that $0 < \alpha < 1$. The interval of convergence of this power series contains $(-1,1)$, as if $a_{k} = (-1)^{k}\prod_{j=0}^{k-1} \frac{\alpha - j}{1 + j}$, then
\begin{align*}
    \abs*{\frac{a_{k+1}x^{k+1}}{a_{k}x^{k}}} = \abs{x} \cdot \frac{\alpha - k}{1 + k} \rightarrow \abs{x} \quad \text{as} \quad k \rightarrow \infty
\end{align*}
gives a radius of convergence of 1 by the Ratio Test. Hence, to apply this power series representation to our values of $x$, we verify that $x = \frac{A_{ij}}{v_{i}v_{j}} \in (-1,1)$, i.e. $\abs{A_{ij}}< v_{i}v_{j}$, for every $i,j$. \\~\\
For $i = j$, we have $A_{ii} \geqslant 0$ so $v_{i}^{2} - A_{ii} > 0$ is gives the desired inequality. For $i \neq j$, observe that $A$ being positive semidefinite means that its principal minors are nonnegative. In particular, for $S = \set{i,j}$, we have
\begin{align*}
    \det(A_{S,S}) = A_{ii}A_{jj} - A_{ij}^{2} \geqslant 0 \implies \abs{A_{ij}} \leqslant \sqrt{A_{ii}A_{jj}} < v_{i}v_{j}
\end{align*}
Having verified that the power series is valid for every entry of our matrix, we have
\begin{align*}
    \nabla^{2}\partial_{T}f_{\alpha} &= \brackets*{v_{i}^{\alpha}v_{j}^{\alpha}}_{ij} \circ \parens*{\bone\bone^{\intercal} - \sum_{k=1}^{\infty} \parens*{\prod_{j=0}^{k-1} \abs*{\frac{\alpha - j}{1 + j}}} \cdot \brackets*{\frac{A_{ij}}{v_{i}v_{j}}}^{\circ k}} \\
    &= \underset{(1)}{\underbrace{\brackets*{v_{i}^{\alpha}v_{j}^{\alpha}}_{ij}}} - \underset{(2)}{\underbrace{\sum_{k=1}^{\infty} \parens*{\prod_{j=0}^{k-1} \abs*{\frac{\alpha - j}{1 + j}}} \cdot \parens*{\brackets*{v_{i}^{\alpha}v_{j}^{\alpha}}_{ij} \circ \brackets*{\frac{A_{ij}}{v_{i}v_{j}}}^{\circ k}}}}
\end{align*}
Here, we recall that $A \circ B$ denotes the Hadamard product of $A,B$, where $(A \circ B)_{ij} = A_{ij}B_{ij}$. Similarly, $A^{\circ k}$ denotes the $k$-iterated Hadamard product of $A$ with itself. \\~\\
All we must do is prove that (1) and (2) are both positive semidefinite, and that (1) is rank-1. Observe that
\begin{align*}
    [v_{i}^{\alpha}v_{j}^{\alpha}]_{ij} = [v_{i}^{\alpha}]_{i} \cdot [v_{i}^{\alpha}]_{i}^{\intercal} \quad\quad \brackets*{\frac{1}{v_{i}v_{j}}}_{ij} = \brackets*{\frac{1}{v_{i}}}_{i} \cdot \brackets*{\frac{1}{v_{i}}}_{i}^{\intercal}
\end{align*}
This tells us (1) is positive semidefinite and rank-1. For the second, observe that
\begin{align*}
    \brackets*{\frac{A_{ij}}{v_{i}v_{j}}}_{ij} = A \circ \brackets*{\frac{1}{v_{i}v_{j}}}_{ij}
\end{align*}
As $A \succcurlyeq 0$ by assumption, this matrix is positive semidefinite by the Schur Product Theorem (see \cref{thm:schur}). Again, inductively applying the Schur Product Theorem, we have $\brackets*{v_{i}^{\alpha}v_{j}^{\alpha}}_{ij} \circ \brackets*{\frac{A_{ij}}{v_{i}v_{j}}}^{\circ k} \succcurlyeq 0$ for every $k$. As (2) is a nonnegative linear combination of positive semidefinite matrices, it is positive semidefinite. 
\end{proof}

\begin{remark}
Not that this operation does not preserve complete log-concavity when $f$ is not assumed to be multiaffine. For example, consider the degree-2 bivariate polynomial $f(x,y) = ax^{2} + bxy + cy^{2}$, where $a,b,c > 0$. Here,
\begin{align*}
    \nabla^{2}f = \begin{bmatrix}
        2a & b \\
        b & 2c
    \end{bmatrix}
\end{align*}
so log-concavity amounts to $\det(\nabla^{2} f) = 4ac - b^{2} \leqslant 0$, i.e. $b^{2} \geqslant 4ac$. Now, raise each coefficient to the power $\alpha$. Then,
\begin{align*}
    \nabla^{2} f_{\alpha} = \begin{bmatrix}
        2a^{\alpha} & b^{\alpha} \\
        b^{\alpha} & 2c^{\alpha}
    \end{bmatrix}
\end{align*}
so log-concavity amounts to $\det(\nabla^{2} f_{\alpha}) = 4a^{\alpha}c^{\alpha} - b^{2\alpha} \leqslant 0$, i.e. $b^{2} \geqslant 4^{1/\alpha}ac$. Clearly, as one decreases $\alpha$ to 0, this inequality gets stronger, which certainly isn't implied by log-concavity of $f$. \\~\\
The problem lies in the fact that when you differentiate a monomial that contains variables with multiplicities, you will obtain ``factorial coefficients'' which are not raised to the power $\alpha$. The operation must be modified appropriately to take this into account. We defer the discussion of the appropriate generalization of this operation for non-multiaffine polynomials to a future article.
\end{remark}

\printbibliography

\end{document}